\documentclass[journal]{IEEEtran}
\usepackage[english]{babel}

\PassOptionsToPackage{hyphens}{url}\usepackage{hyperref}

\usepackage{algorithmic}
\usepackage[ruled,vlined]{algorithm2e}
\usepackage{amsfonts}
\usepackage{amsmath}
\usepackage{amsthm}
\usepackage{amssymb}
\usepackage{bbm}
\usepackage{cite}
\usepackage[table]{xcolor}
\usepackage[mathscr]{euscript}
\usepackage{enumitem}
\usepackage{epstopdf}
\usepackage{float}
\usepackage[T1]{fontenc}
\usepackage{glossaries}
\usepackage{graphicx}
\usepackage[latin9]{inputenc}
\usepackage{times}
\usepackage{latexsym}
\usepackage{mathtools}
\usepackage{multirow}
\usepackage{makeidx}
\usepackage{mathrsfs} 
\usepackage{nomencl}
\usepackage{rotating}
\usepackage{textcomp}
\usepackage{varioref}
\usepackage{xcolor}
\usepackage{phonetic}

\ifCLASSOPTIONcompsoc
\usepackage[caption=false,font=normalsize,labelfon
t=sf,textfont=sf]{subfig}
\else
\usepackage[caption=false,font=footnotesize]{subfi
g}
\fi

\DeclareGraphicsExtensions{.eps,.png}

\allowdisplaybreaks
\newcommand{\bpara}[1]		{\medskip \noindent {\bf #1}}

\onecolumn

\usepackage{soul,hyperref}
\usepackage[scaled=.80]{beramono}

\usepackage{cite}
\usepackage[switch]{lineno}

\ifCLASSINFOpdf
\else
\fi

\begin{document}
\renewcommand{\thefootnote}{\normalsize \arabic{footnote}} 	
\newtheorem{theo}{Theorem}
\newtheorem{rem}{Remark}
\newtheorem{conj}{Conjecture}
\newtheorem{definition}{Definition}
\newtheorem{prop}{Proposition}
\newtheorem{lem}{Lemma}
\newtheorem{cor}{Corollary}
\newtheorem{remark}{Remark}
\newtheorem{example}{Example}

\newcommand{\pw}{PW_{\Omega}}
\newcommand{\kz}{k\in\mathbb{Z}}
\newcommand{\akz}{\forall k \in \mathbb{Z}}
\newcommand{\tu}{\mathcal{T}_u}
\newcommand{\db}{\bar{\delta}}
\newcommand{\Z}{\mathbb{Z}}
\newcommand{\N}{\mathbb{N}}
\newcommand{\R}{\mathbb{R}}
\newcommand{\LR}{L^2(\mathbb{R})}
\newcommand{\LO}{L^2([-\Omega,\Omega])}
\newcommand{\re}{\mathbb{R}}
\newcommand{\co}{\mathbb{C}}
\newcommand{\cH}{\mathcal{H}}
\newcommand{\Tu}{\mathcal{T}_u}
\newcommand{\F}{\mathcal{F}}
\newcommand{\bs}{\boldsymbol}
\newcommand{\eq}{\triangleq}
\newcommand{\cond}{S}
\newcommand{\tM}{t_{\mathsf{M}}}
\newcommand{\KN}{\mathcal{K}_N}
\newcommand{\Sz}{\mathbb{S}_N}
\newcommand{\s}{ISI}
\newcommand{\MO}[0]{\mathscr{M}_\lambda}
\newcommand{\fig}[1]{Fig.~\ref{#1}}
\newcommand{\MOh}[0]{\mathscr{M}_{\boldsymbol{\mathsf{H}}}}
\newcommand{\IFint}[0]{\mathcal{L}_k}
\newcommand{\IF}[0]{\mathrm{ASDM}_\theta}
\newcommand{\EQc}[1]		{\stackrel{(\ref{#1})}{=}}
\newcommand{\setsep}{\ \big\vert\ }
\newcommand{\fe}[1]{\left[\kern-0.30em\left[#1  \right]\kern-0.30em\right]}

\newcommand{\w}[1]{\widetilde{#1}}

\newcommand{\Ln}{\mathcal{L}_n}

\newcommand{\bLn}{\overline{\mathcal{L}}_n}

\newcommand{\LN}[1]{\mathcal{L}_{#1}}
\newcommand{\etab}{\eta_\delta}

\newcommand{\fpM}{f'_{\mathsf{M}}}
\newcommand{\fppM}{f''_{\mathsf{M}}}
\newcommand{\fpm}{f'_{\mathsf{m}}}
\newcommand{\fpMb}{\varphi'_{\mathsf{M},\delta}}
\newcommand{\fpmb}{\varphi'_{\mathsf{m},\delta}}
\newcommand{\fppMb}{\varphi''_{\mathsf{M},\delta}}

\newcommand{\fpMs}[1]{f_{\mathsf{M},#1}'}
\newcommand{\fpms}[1]{f_{\mathsf{m},#1}'}

\newcommand{\dtM}{\Delta t_{\mathsf{M}}}
\newcommand{\dtm}{\Delta t_{\mathsf{m}}}
\newcommand{\aM}{a_{\mathsf{M}}}

\newcommand{\FD}[3]{\mathscr{D}_{#1}^{#2}\left[#3\right]}
\newcommand{\NFD}{\mathscr{D}}
\newcommand{\syntharg}[1]{\mathcal{S}_\Omega\sqb{{#1}}}
\newcommand{\synth}{\mathcal{S}_\Omega}

\newcommand{\locav}{\mathcal{L}}

\newcommand{\diff}{d}

\newcommand{\FS}[3]{\mathscr{F}\rb{{#1}_{k}^{#3},#2_{k}}}

\newcommand{\vb}[1]{\left\lvert #1 \right\rvert}
\newcommand{\rb}[1]{\left( #1 \right)}
\newcommand{\sqb}[1]{\left[ #1 \right]}
\newcommand{\cb}[1]{\left\lbrace #1 \right\rbrace}
\newcommand{\floor}[1]{\left\lfloor #1 \right\rfloor}
\newcommand{\ceil}[1]{\left\lceil #1 \right\rceil}

\newcommand\ab[1]			{{\color{blue}#1}}

\newcommand{\df}[1]{{\color{red} #1}}

\newcommand\abb[1]			{{\color{red}#1}}

\newcommand\ETP[1]         {\mathbb{E}_{p}\left(#1\right)}

\newcommand\ETPP[2]         {\mathbb{E}_{#1}\left(#2\right)}

\renewcommand\tilde{\widetilde}

\DeclarePairedDelimiter{\norm}{\Vert}{\Vert}
\DeclarePairedDelimiter{\abs}{\left|}{\right|}
\DeclarePairedDelimiter{\Prod}{\langle}{\rangle}

\newcommand{\PW}[1]{\mathsf{PW}_{#1}}

\newcommand{\M}[3]{\vb{\mu_{#1}^{#2}\sqb{#3}}}
\newcommand{\B}[3]{\vb{\beta_{#1}^{#2}\sqb{#3}}}

\newcommand{\Tm}{T_{\mathsf{m}}}
\newcommand{\TM}{T_{\mathsf{M}}}

\newcommand{\wTm}{{T}_{\mathsf{m}}}
\newcommand{\wTM}{{T}_{\mathsf{M}}}

\newcommand{\iPW}[2]{#1 \in \mathsf{PW}_{#2}}

\renewcommand\geq\geqslant \renewcommand\leq\leqslant
\newcommand{\const}[1]{#1}

\def\ind{1}

\def\th{\Psi}

\def\figmode    {1}      
\def\PH         {0}        
\def\stabilization	{0}
\def\True		{1}

\def\TR                 {1}
\def\FL                 {0}

\def\BoundDerInAppendix {1}
\def\NoisyRecInAppendix {1}

\title{A Generalized Approach for Recovering Time Encoded Signals with Finite Rate of Innovation}

\author{Dorian~Florescu
	
\thanks{This work is supported by the UK Research and Innovation (UKRI) council's \emph{Future Leaders Fellowship} program ``Sensing Beyond Barriers'' (MRC Fellowship award no.~MR/S034897/1).}	
\thanks{D.~Florescu is with the Department of Electrical and Electronic Engineering, Imperial College London, SW72AZ, UK. E-mail: D.Florescu@imperial.ac.uk or fdorian88@gmail.com.}
\thanks{Manuscript received Feb. XX, 20XX revised August XX, 20XX.}
}

\markboth{Journal of \LaTeX\ Class Files,~Vol.~XX, No.~X, August~20XX}
{Shell \MakeLowercase{\textit{et al.}}: Bare Demo of IEEEtran.cls for IEEE Journals}

\maketitle

\begin{abstract}
In this paper, we consider the problem of recovering a sum of filtered Diracs, representing an input with finite rate of innovation (FRI), from its corresponding time encoding machine (TEM) measurements. So far, the recovery was  guaranteed for cases where the filter is selected from a number of particular mathematical functions. Here, we introduce a new generalized method for recovering FRI signals from the TEM output. On the theoretical front, we significantly increase the class of filters for which reconstruction is guaranteed, and provide a condition for perfect input recovery depending on the first two local derivatives of the filter. We extend this result with reconstruction guarantees in the case of noise corrupted FRI signals. On the practical front, in cases where the filter has an unknown mathematical function, the proposed method streamlines the recovery process by bypassing the filter modelling stage. We validate the proposed method via numerical simulations with filters previously used in the literature, as well as filters that are not compatible with the existing results.  Additionally, we validate the results using a TEM hardware implementation.
\end{abstract}


\IEEEpeerreviewmaketitle

\section{Introduction}
\label{sec:intro}

\IEEEPARstart{S}{hannon's} iconic work on reconstructing bandlimited signals from uniform samples \cite{Shannon:1949:C} was generalized on multiple levels. A notable generalization is from bandlimited signals to signals belonging to shift-invariant spaces (SIS), which enables reconstructing a linear combination of uniformly spaced filters $g(t)=\sum_{k\in\Z} a_k \varphi(t-kT)$ from uniform samples \cite{Unser:2000:C} and nonuniform samples \cite{Aldroubi:2001:J}. An important advantage of SIS is that the theory is backward compatible with previous theory when $\varphi(t)$ is a sinc function, or a B-spline, but also enables choosing new functions $\varphi(t)$ satisfying some predefined requirements \cite{Unser:2000:C}. The fact that the SIS filters are uniformly spaced can be a strong constraint when the input is a sparse sequence of filters centered in real-values $\cb{\tau_k}_{k=1}^K$:
\begin{equation}
\label{eq:FRI_signal}
    g(t)=\sum\nolimits_{k=1}^K a_k\varphi(t-\tau_k), \quad t\in\sqb{0,\tM}.
\end{equation}
Such signals can no longer be considered part of a SIS. This recovery problem has twice the number of unknowns as before, requiring to compute $\cb{\tau_k,a_k}_{k=1}^K$ using measurements of $g(t)$. We call this problem finite-rate-of-innovation (FRI) sampling \cite{Vetterli:2002}.
The versatility of FRI sampling allowed its application in a large number of areas such as ECG acquisition and compression \cite{Baechler:2017:J}, radioastronomy \cite{Pan:2016:J}, image processing \cite{Wei:2016:J}, ultrasound imaging \cite{Tur:2011:J}, calcium imaging \cite{Onativia:2013:J}, or the Unlimited Sensing Framework \cite{Bhandari:2020:J,Bhandari:2021:Ja,Florescu:2021:J,Florescu:2022:Jb}.

In this paper we consider the problem of recovering \eqref{eq:FRI_signal} from Time Encoding Machine (TEM) measurements, which represents a different generalization of Shannon's sampling inspired from the information processing in the brain, characterized by low power consumption \cite{Lazar:2004:Jc}. A TEM with input $g(t)$ is an operator $\mathcal{T}$ defined as $\mathcal{T}g=\cb{t_k}_{k\in\mathbb{Z}}$, where $\cb{t_k}$ is a strictly increasing sequence of time samples known as spikes or trigger times. Input recovery was demonstrated for the case when $g(t)$ is a bandlimited function \cite{Lazar:2004:Jc,Florescu:2015:J}, a function in a shift-invariant space \cite{Gontier:2014:J,Florescu:2015:J}, or a bandlimited function with jump discontinuities \cite{Florescu:2021:C,Florescu:2022:Jb}.

\bpara{Related Work.} The work on FRI signal recovery from TEM measurements started with \cite{Alexandru:2019:J} and is still at an early stage. Furthermore, currently, the input recovery guarantees assume that the FRI filters $\varphi(t)$ are particular mathematical functions such as polynomial or exponential splines \cite{Alexandru:2019:J}, hyperbolic secant kernels \cite{Hilton:2023:C} or alpha synaptic functions \cite{Hilton:2021:Cb}. The case of periodic FRI signals was considered in \cite{Rudresh:2020:C,Namman:2021:J,Kalra:2022:C,Kamath:2023:C}. We note that this line of work assumes that the TEM input is bandlimited and uses Nyquist rate type conditions for recovery. A common scenario is that $\varphi(t)$ is not chosen by design, but rather results from the physical properties of the acquisition device \cite{Baechler:2017:J}. In such cases, the existing work on FRI signal recovery for TEMs does not offer any guarantees. Moreover, the existing work is based on exact analytical tools for recovery, such as Prony's method, which are known to not allow good stability to noise or model mismatch. The study of general filters $\varphi(t)$ was first numerically validated in \cite{Florescu:2023:Ca}.

\bpara{Contributions.} The contributions are as follows.
\begin{enumerate}
    \item We introduce a new general FRI signal recovery method from TEM samples that is compatible with the setups of the existing methods.
    \item We show theoretically that, in its full generality, the proposed method perfectly recovers $g(t)$ from the TEM output for significantly reduced constraints on $\varphi(t)$. We show that the constraints are achievable for high enough sampling rates.
    \item We introduce noise robustness recovery guarantees.
    \item We validate the new method using various filters including a randomly generated filter.
    \item We test our method using TEM hardware measurements.
\end{enumerate}

\bpara{Organization.} The paper is organized as follows. Section \ref{sect:tem} gives the preliminaries of TEMs. Section \ref{sect:FRI_recovery} introduces the sampling setup and summarises the existing recovery methods. The proposed method is in Section \ref{sect:prop_method}. The noisy scenario is considered in Section \ref{sect:noise}. Section \ref{sect:numerical_study} presents numerical and hardware experiments. The proofs to some of the results are in Section \ref{sect:proofs}, and the conclusions are in Section \ref{sect:conclusions}.

\begin{table}[!h]
\caption{Frequently Used Symbols}
\rowcolors{2}{blue!4}{white}
\renewcommand{\arraystretch}{1}
\begin{center}
\begin{tabular}{m{1.4cm}	 || m{6.7cm} l}
	Symbol & Definition \\ \hline  \hline  
	$\varphi(t)$ & Kernel for FRI input, with finite time support $[-L,\infty)$.\\
        $\cb{\tau_k,a_k}_{k=1}^K$ & Time locations and amplitudes of the FRI input.\\
        $\varepsilon_\tau, \varepsilon_a$ & Minimum separation and amplitude of input pulses.\\
        $f(t)$ & Filter $\varphi(t)$ translated in $\tau_1$, i.e., $f(t)=\varphi(t-\tau_1)$\\
        $g(t), g_\infty$ & FRI input and an upper bound  $g_\infty\geq\vb{g(t)}$.\\
        $\widehat{g}(t)$ & The input estimation for $t\in\sqb{0,t_{\mathsf{m}}}$.\\
        $\delta,b$ & TEM threshold and bias parameters, respectively.\\
        $\cb{t_k}_{k=1}^N$ & Output switching times of the TEM.\\
        $\Tm,\TM$ & Min/max of distance between $t_k$: $\Tm\leq\Delta t_k \leq \TM$.\\
        $n_k$ & Index of $t_{n_k}$ located right after the onset of the $k$th filter.\\
        $t_\delta$ & Max distance between the onset of the $k$th filter and $t_{n_k+1}$.\\        
        $\fpm,\fpM,\fppM$ & Min/max of $f'(t)$ and $f''(t)$ for $t\in\sqb{t_{n_k+1},t_{n_k+3}}$.\\
        $\fpmb,\fpMb$ & Min/max of $\varphi'(t)$ for $t\in\sqb{-L+\Tm/2,-L+4\TM}$.\\
        $\varphi_-'(t), \varphi_+'(t)$ & Left and right derivatives of $\varphi(t)$.\\
        $e_{\tau}, e_a$ & Error bounds for estimating $\tau_k$ and $a_k$, respectively.\\
        $\mathcal{L}_n$, $I_n(\tau)$ & $\mathcal{L}_n g=\int_{t_n}^{t_{n+1}} g(s)ds$, $I_n(\tau)=\mathcal{L}_n \varphi(\cdot-\tau)$.\\
        $\epsilon$ & Defines the slope (derivative) of ${\frac{I_{n+2}(\tau)}{I_{n+1}(\tau)}}$.\\
        $\norm{f}_\infty$ & The absolute norm $\norm{f}_\infty=\sup_{t\in\R}\vb{f(t)}$.\\
        $C(S)$ & The set of continuous functions on set $S$.\\ 
        $C^N(S)$ & The set of $N$th order differentiable functions on set $S$.\\ 	

\end{tabular}
\end{center}
\label{tab:my-table}
\end{table}

\section{The Time Encoding Machine}
\label{sect:tem}

We consider two classes of TEMs: the Asynchronous Sigma-Delta Modulator (ASDM) and the integrate-and-fire (IF). The ASDM is characterised by low power consumption and modular design \cite{Ozols:2015:C}, comprising a loop with an adder, integrator, and a noninverting Schmitt trigger, as depicted in \fig{fig:asdm}(a).
\begin{figure}[!t]
	\includegraphics[trim={0cm 0cm 0cm 0},clip,width=9cm]{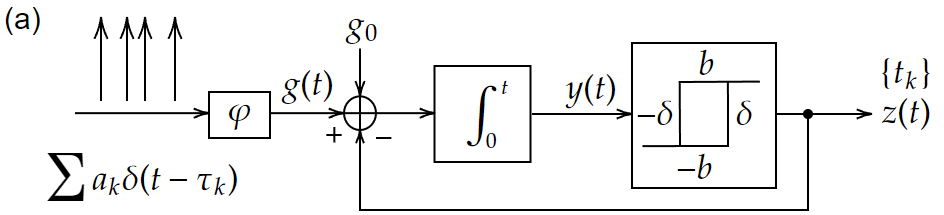}
    \includegraphics[trim={0cm 0cm 0cm 0},clip,width=9cm]{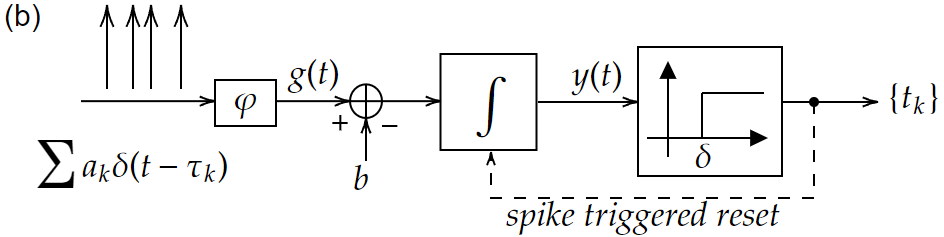}
	\caption{The TEMs considered in this work: (a) The asynchronous sigma-delta modulator (ASDM). (b) The integrate-and-fire (IF) model.}
	\label{fig:asdm}
\end{figure}
The initial conditions are $z(0)=-b$ and $y(0)=0$, where $b$ is a positive constant. We assume that $\vb{g(t)+g_0}<b$, which ensures that $y(t)$ is strictly increasing in the immediate positive vicinity of $t=0$. This means that eventually $y(t)=\delta$, and this time point represents the first ASDM output sample  $t=t_1$. This determines the ASDM output to change to $z(t)=b$, which, in turn, ensures that $y(t)$ is strictly decreasing for $t>t_1$. Eventually $y(t)=-\delta$ for $t=t_2$, $z(t)$ toggles back to $-b$, and the process continues recursively. The output sequence $\cb{t_n}_{n\geq1}$ satisfies the \emph{t-transform} equations\cite{Lazar:2004:Jc}
\begin{equation}
\Ln g= \rb{-1}^n\sqb{2\delta - b\Delta t_{n}}-g_0 \Delta t_{n},\quad n\in\Z_+^*,
\label{eq:ttransform}
\end{equation}
where $\Delta t_n\triangleq t_{n+1}-t_n$,  $\mathcal{L}_ng\triangleq \int_{t_n}^{t_{n+1}} g\rb{s} $, $\delta$ and $b$ are the threshold and amplitude of the Schmitt trigger output, respectively. Parameter $g_0$ represents a bias that is typically considered $0$ in simulations, but plays a role in explaining hardware measurements \cite{Florescu:2022:C}. 

The IF TEM is inspired from neuroscience, previously used to fit models of biological neurons \cite{Florescu:2018:J}, but also used in machine learning \cite{Maass:2002:J,Florescu:2019:J} or recovery of FRI signals \cite{Alexandru:2019:J,Hilton:2023:C}. The functioning principle of the IF, depicted in \fig{fig:asdm}(b), is as follows. The input $g(t)$, added with bias parameter $b$, is integrated, which results in strictly increasing function $y(t)$. Each time $y(t)$ crosses threshold $\delta$, the integrator is reset and the IF generates an output spike time $t_k$. The IF TEM is described by the following equations
\begin{equation}
    \Ln g=\delta-b\Delta t_n,\quad n\in\Z_+^\ast.
    \label{eq:ttransformIF}
\end{equation}

For both the ASDM and IF the assumption is that $\vb{g(t)}\leq g_\infty <b$, leading to the following sampling density bounds 
\begin{equation}
\label{eq:bound_isi}
    T_{\mathsf{m}}\leq \Delta t_n \leq T_{\mathsf{M}},
\end{equation}
where $T_{\mathsf{m}}\triangleq\frac{2\delta}{b+g_\infty}$, $T_{\mathsf{M}}\triangleq\frac{2\delta}{b-g_\infty}$ for the ASDM and $T_{\mathsf{m}}\triangleq\frac{\delta}{b+g_\infty}$, $T_{\mathsf{M}}\triangleq\frac{\delta}{b-g_\infty}$ in the case of the IF \cite{Lazar:2004:Jc}.

\vspace{2em}
\section{Recovery of FRI Signals from TEM Samples}
\label{sect:FRI_recovery}
\subsection{Proposed Sampling Setup}

Let  $g(t)$ belong to the input space spanned by \eqref{eq:FRI_signal},
where $\cb{a_k,\tau_k}_{k=1}^K$ are unknown values satisfying $\varepsilon_a<\vb{a_k}<g_\infty, 0<\tau_k<\tau_{k+1}<\tM, \tau_{k+1}-\tau_k>\varepsilon_\tau$, $\varepsilon_a,\varepsilon_\tau,g_\infty,\tM>0$ and $K$ is the unknown number of pulses with shape $\varphi(t)$. We also assume that $\varphi(t)$ is known, $\varphi\in  C^2(-L,0){\cap C\rb{\R}}$, and that $ \mathrm{supp} \rb{\varphi} \subseteq [-L,\infty)$.
In other words, $\varphi(t)$ has finite support, is second order differentiable on $(-L,0)$ and continuous on $\R$. Furthermore, we assume that the left derivative $\varphi_-'(t)$ and right derivative $\varphi_+'(t)$ exist and are bounded for $t\in\R$. We also assume that $\varphi'(t)>0, t\in(-L,0)$.
Furthermore, we assume that $\max_t\vb{\varphi(t)}=1$, which does not reduce generality. The conditions on $\varphi(t)$ are defining a space of functions that are relatively common, including the previously studied cases of polynomial and exponential splines \cite{Alexandru:2019:J} and alpha synaptic activation functions \cite{Hilton:2021:Cb}. The hyperbolic secant kernel \cite{Hilton:2023:C} does not fully satisfy the conditions as its support is $\R$, but our analysis is also applicable given its fast decay to $0$ (see Section \ref{sect:numerical_study}).

Furthermore, we assume $\max_t\vb{g(t)}\leq g_\infty$, $\varepsilon_{\tau}<L$, $\tau_1\geq L$, and that $g(t)$ is sampled with an ASDM or IF TEM over the finite time interval $\sqb{0,\tM}$ to yield output time encoded samples $\cb{t_n}_{n=1}^N$. Both TEMs enable computing $\cb{\Ln g}_{n=1}^{N-1}$ from $\cb{t_n}_{n=1}^N$ via \eqref{eq:ttransform}, \eqref{eq:ttransformIF}, respectively. The problem we propose is to recover $\cb{a_k,\tau_k}_{k=1}^K$ from $\cb{t_n}_{n=1}^N$, i.e., to compute $\cb{a_k,\tau_k}_{k=1}^K$ from $\cb{\Ln g}_{n=1}^{N-1}$, which is independent of the particular TEM model used.  We assume that $\varepsilon_\tau>4\TM$, which ensures there are at least $4$ TEM samples in between each two consecutive pulses via \eqref{eq:bound_isi}. 

\subsection{Existing Recovery Methods}
\label{sect:prev_methods}

The work in \cite{Alexandru:2019:J} considers the estimation of $\cb{\tau_k,a_k}$ from $\cb{t_n}$ for an IF TEM where the filter $\varphi(t)$ is a polynomial or exponential spline (E-spline), compactly supported with support length $L$. It is assumed that the pulses have no overlaps, i.e., $\tau_{k+1}-\tau_k>L$. Moreover, for identifying pulse $a_k \varphi\rb{t-\tau_k}$, three spike times $\cb{t_{n+i}}_{i=0}^2$ are used, which are assumed to be located in an interval of length $L/2$ at the onset of the pulse. Furthermore, it is assumed that 
\begin{equation}
    \sum_{n=1}^2 c_{n,m}\rb{\varphi\ast\ind_{[0,t_{n+1}-t_n]}(t-\tau_n)}=e^{\jmath \omega_m t}, \quad m\in\cb{0,1},
\end{equation}
has exact analytical solutions $c_{n,m}$, where $\omega_0,\omega_1$ are parameters of the E-spline. The values of $\tau_k$ and $a_k$ are then found by computing the signal moments $s_m=\sum_{n=1}^2 c_{n,m} y(t_{n+1})$ and then solving $s_m=a_k e^{\jmath \omega_m \tau_k}$ via Prony's method. These results are extended for inputs generated with polynomial splines and piece-wise constant signals. When the pulses overlap, i.e., $\tau_{k+1}-\tau_k<L$, recovery is still possible if samples of multiple TEMs are recorded \cite{Alexandru:2019:J}. 
For a filter $\varphi(t)=\mathrm{sech}^2(t)$  \cite{Hilton:2023:C}
\begin{equation}
\label{eq:sech}
    \Ln g=\frac{e^{t_n}P\rb{e^{2t_n}}}{Q\rb{e^{2t_n}}},\quad Q(x)=\sum_{k=1}^K\rb{1+e^{-2\tau_k}x},
\end{equation}
where $P(x)$ and $Q(x)$ are polynomials. The values $\cb{\tau_k,a_k}$ can then be uniquely recovered by solving analytically \eqref{eq:sech} via \eqref{eq:ttransformIF}, where the recovery approach is inspired from the recovery of FRI signals from nonuniform samples \cite{Wei:2012:J}. Moreover, a different analytical recovery was shown for $\varphi(t)=te^{-t}\cdot\ind_{[0,\infty)}(t)$. We note that the functions $\varphi(t)$ above are piecewise elementary functions. Elementary functions represent a small subset of all continuous functions \cite{Geddes:1992:B}. In fact, in practice, the impulse response of a filter rarely fits perfectly a mathematical expression, as it often results from the physical properties of a given acquisition device \cite{Shukla:2007:J}. Therefore, using the methods above may introduce an additional error source due to model mismatch. Thus, introducing a new method allowing perfect recovery for general filters $\varphi(t)$ could tackle significantly wider scenarios and applications.

A separate line of work considers the case where $g(t)$ is a periodic bandlimited signal, and the TEM sampling rate satisfies a Nyquist rate type recovery guarantee, leading to  parametric input recovery approaches \cite{Rudresh:2020:C,Namman:2021:J,Kamath:2023:C,Kalra:2022:C}. In this work, we consider the general scenario where the input to the TEM is apriodic and not necessarily bandlimited.

\vspace{2em}
\section{The Proposed Recovery Method}
\label{sect:prop_method}

As discussed in Section \ref{sect:prev_methods}, assuming that $\varphi(t)$ is an E-spline, the method in \cite{Alexandru:2019:J} requires three consecutive TEM samples $\cb{t_{n+i}}_{i=0}^2$ located at the onset of the pulse to be estimated, which amounts to two consecutive integrals $\cb{\LN{n+i} g}_{i=0}^1$. Here we will show that this information is enough to recover the pulse when $\varphi(t)$ satisfies much more relaxed assumptions, which don't require that the pulse is generated with any particular mathematical function (e.g. exponential, or polynomial). 

\subsection{The Case of One Pulse}

We first assume that $K=1$, and subsequently extend to $K\geq 1$. Let $n_1\in\Z$ be the index of the TEM output located right after the onset of filter $\varphi(t-\tau_1)$, defined as
\begin{equation}
\label{eq:n1}
n_1\triangleq\min_{n\in\cb{1,\dots,N}} \cb{n \setsep t_n>\tau_1-L}.    
\end{equation}
Our recovery makes use of the TEM output samples $t_{n_1+i}, i\in\cb{1,2,3}$, which, assuming $4\TM<L$, satisfy
\begin{equation}
    0<t_{n_1-1}\leq\tau_1-L<t_{n_1+i}<\tau_1,\quad i\in\cb{0,\dots,3}.
    \label{eq:four_spikes}
\end{equation} 
Let $I_n(\tau)\triangleq \int_{t_n}^{t_{n+1}} \varphi(t-\tau) dt $. The idea behind the recovery is to compute the ratio of two consecutive integrals
\begin{equation}    \frac{\mathcal{L}_{n_1+2}g}{\mathcal{L}_{n_1+1}g}=\frac{a_1 I_{n_1+2}(\tau_1)}{a_1 I_{n_1+1}(\tau_1)}=\frac{I_{n_1+2}(\tau_1)}{I_{n_1+1}(\tau_1)},
\label{eq:tau1_computation}
\end{equation}
which is not a function of $a_1$, but only $\tau_1$. If  $\frac{I_{n_1+2}(\tau)}{I_{n_1+1}(\tau)}$ is strictly increasing, then $\tau_1$ can be uniquely estimated from \eqref{eq:tau1_computation}. This recovery approach is illustrated in \fig{fig:rec_diagram}. 
\begin{figure}[!t]
	\centerline{\includegraphics[trim={0cm 0cm 0cm 0},clip,width=9cm]{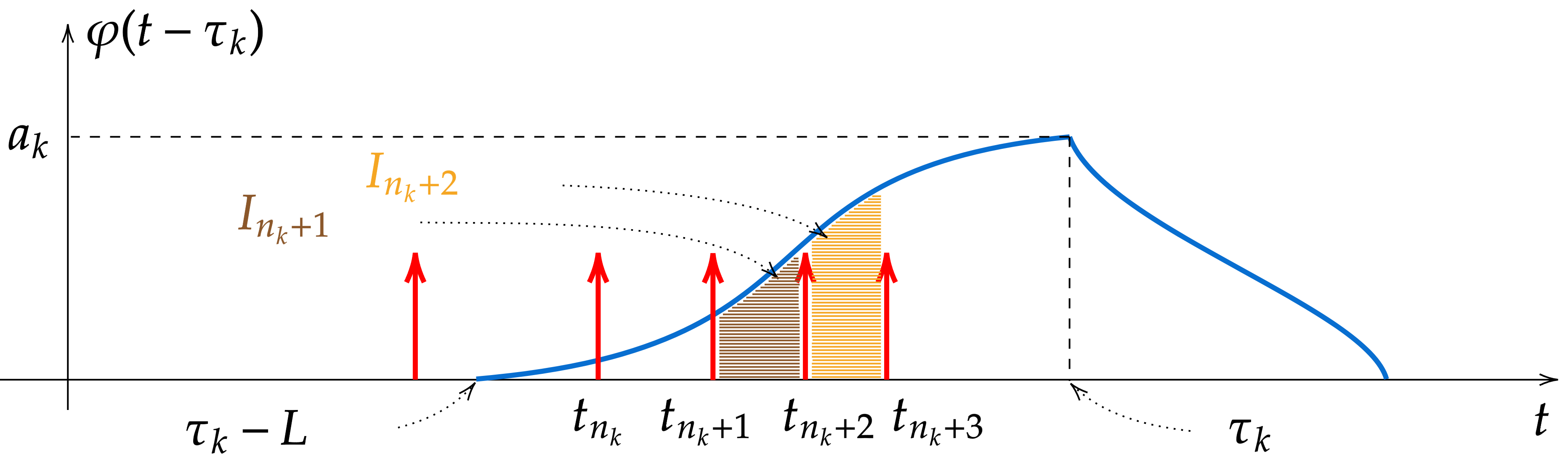}}
	\caption{The recovery of the $k$th pulse from the FRI input. The integrals $I_{n_k+1}$ and $I_{n_k+2}$, computed from the TEM output $\cb{t_{n_k+i}}_{i=0}^3$, contain all the information needed to uniquely identify $\tau_k$ and $a_k$.}
	\label{fig:rec_diagram}
\end{figure} 
The following lemma derives conditions to guarantee the required monotonicity.
\begin{lem}
\label{lem:monotonic}
Function $\frac{I_{n_1+2}(\tau)}{I_{n_1+1}(\tau)}$ is finite, differentiable, and
\begin{equation}
\label{eq:monotonicity}
\begin{gathered}
\rb{\frac{I_{n_1+2}(\tau)}{I_{n_1+1}(\tau)}}'\geq \frac{{\fpmb}^3}{2\fpMb}\frac{\Tm^3}{\TM^2}\cdot\bar{\epsilon}_{\mathsf{m}},\\
\bar{\epsilon}_{\mathsf{m}}\triangleq2+\frac{\Tm}{\TM}-\frac{2\varphi\rb{t_\delta}}{{\fpmb}^2}\cdot\fppMb\cdot\frac{2\TM}{\Tm}-\frac{{\fpMb}^2}{{\fpmb}^2},
\end{gathered}
\end{equation}
if $4\TM<L$, where $t_\delta=-L+2\TM, \fpMb\triangleq\max_{t\in\mathbb{S}_{\delta} }\varphi'(t), \fpmb\triangleq\min_{t\in\mathbb{S}_{\delta}} \varphi'(t), \fppMb\triangleq \max_{t\in\mathbb{S}_{\delta}} \varphi''(t)$, and  $\mathbb{S}_{\delta}\triangleq \sqb{-L+\tfrac{\Tm}{2},-L+4\TM}$.
\end{lem}

\if\BoundDerInAppendix\FL
\begin{proof}
The filter satisfies $\varphi(t)=0,t<-L$ and, given that $\varphi'(t)>0, t\in (-L,0)$, it follows that $\varphi(t)>0, t\in (-L,0)$ and thus $\varphi(t-\tau)>0, t\in (-L+\tau,\tau)$. 
Given that $4\TM<L$, we get \eqref{eq:four_spikes}, and therefore  $I_{n_1+1}(\tau)>0$ and thus $\frac{I_{n_1+2}(\tau)}{I_{n_1+1}(\tau)}$ is well-defined. Moreover, the ratio is the composition of differentiable functions, therefore it is itself differentiable. 

For simplicity, let $f(t)\triangleq \varphi(t-\tau), t \in\sqb{\tau-L,\tau}, f_l=f\rb{t_l}, l\in\cb{n_1+1,n_1+2,n_1+3}$. The following holds
\begin{align*}    I_{n_1+i}'(\tau)&=\varphi\rb{t_{n_1+i}-\tau}-\varphi\rb{t_{n_1+i+1}-\tau}=-\Delta f_{n_1+i},
\end{align*}
for $i\in\cb{1,2}$. Using the above, the following holds
\begin{equation}
\label{eq:der1}
\begin{aligned}
    \rb{\frac{I_{n_1+2}(\tau)}{I_{n_1+1}(\tau)}}'= \frac{I_{n_1+2}'(\tau)I_{n_1+1}(\tau)-I_{n_1+2}(\tau)I_{n_1+1}'(\tau)}{I_{n_1+1}^2(\tau)}
    =\frac{\Delta f_{n_1+1} I_{n_1+2}(\tau)-\Delta f_{n_1+2}  I_{n_1+1}(\tau)}{I_{n_1+1}^2(\tau)}\triangleq \epsilon.
\end{aligned}
\end{equation}
The final objective is to show that $\epsilon$ is positive. We proceed by providing subsequent lower bounds for $\epsilon$ as follows.
By rearranging the terms in \eqref{eq:der1} and using that $I_n(\tau)=\Ln f$,
\begin{equation}
    \frac{\epsilon\cdot I_{n_1+1}^2(\tau)}{\Delta f_{n_1+1} \Delta f_{n_1+2}}=\frac{\LN{n_1+2}f}{\Delta f_{n_1+2}}-\frac{\LN{n_1+1}f}{\Delta f_{n_1+1}}.
    \label{eq:der2}
\end{equation}
Function $f(t)$ is positive, differentiable and strictly increasing just as $\varphi(t)$. We next expand $f(t)$ in Taylor series with anchor points $t_{n_1+1}$ and $t_{n_1+2}$, respectively,  $$f(t)=f_{n_1+1}+f'\rb{\xi_{n_1+1}}(t-t_{n_1+1})\leq f_{n_1+1} +\fpM(t-t_{n_1+1}),$$ $$f(s)=f_{n_1+2}+f'\rb{\xi_{n_1+2}}(s-t_{n_1+2})\geq f_{n_1+2} +\fpm(s-t_{n_1+2}),$$ 
for $t\in\sqb{t_{n_1+1},t_{n_1+2}}$ and $s\in\sqb{t_{n_1+2},t_{n_1+3}}$, such that 
\begin{gather*}
t_{n_1+1}\leq \xi_{n_1+1}\leq t\leq t_{n_1+2}, \quad\fpM\triangleq\max_{t\in\sqb{t_{n_1+1},t_{n_1+3}}}f'(t)\\
t_{n_1+2}\leq \xi_{n_1+2}\leq s\leq t_{n_1+3}, \quad\fpm\triangleq\min_{t\in\sqb{t_{n_1+1},t_{n_1+3}}}f'(t).
\end{gather*}
We can thus bound the local integrals of $f(t)$ as 
\begin{equation}
\begin{gathered}
    \LN{n_1+1}f \leq f_{n_1+1}\cdot\Delta t_{n_1+1} + \fpM \frac{\Delta t_{n_1+1}^2}{2},\\
    \LN{n_1+2}f \geq f_{n_1+2}\cdot\Delta t_{n_1+2} + \fpm \frac{\Delta t_{n_1+2}^2}{2}.   
\end{gathered}
\label{eq:der2_2}
\end{equation}
By combining \eqref{eq:der2} and \eqref{eq:der2_2} we get
\begin{align}
\label{eq:der3}
\begin{aligned}
\frac{\epsilon\cdot I_{n_1+1}^2(\tau)}{\Delta f_{n_1+1} \Delta f_{n_1+2}}>f_{n_1+2}\frac{\Delta t_{n_1+2}}{\Delta f_{n_1+2}}+\frac{\fpm}{2}\frac{\Delta t_{n_1+2}^2}{\Delta f_{n_1+2}}
-\rb{f_{n_1+1}\frac{\Delta t_{n_1+1}}{\Delta f_{n_1+1}}+\frac{\fpM}{2}\frac{\Delta t_{n_1+1}^2}{\Delta f_{n_1+1}}}.
\end{aligned}
\end{align}
We then use that $f_{n_1+1},f_{n_1+2}>0$, and thus
\begin{equation}
\label{eq:der4}
\begin{gathered}
    f_{n_1+1}\frac{\Delta t_{n_1+1}}{\Delta f_{n_1+1}}+\frac{\fpM}{2}\frac{\Delta t_{n_1+1}^2}{\Delta f_{n_1+1}}\leq \frac{f_{n_1+1}}{\fpm}+\frac{\fpM\Delta t_{n_1+1}}{2\fpm},\\
    f_{n_1+2}\frac{\Delta t_{n_1+2}}{\Delta f_{n_1+2}}+\frac{\fpm}{2}\frac{\Delta t_{n_1+2}^2}{\Delta f_{n_1+2}}\geq \frac{f_{n_1+2}}{\fpM}+\frac{\fpm\Delta t_{n_1+2}}{2\fpM}.
\end{gathered}
\end{equation}
As before, we plug \eqref{eq:der4} into  \eqref{eq:der3} and get
\begin{equation*}
\frac{\bar{\epsilon}\Delta t_{n_1+1}\fpm}{2\fpM}>\frac{f_{n_1+2}}{\fpM}+\frac{\fpm\Delta t_{n_1+2}}{2\fpM}
-\rb{\frac{f_{n_1+1}}{\fpm}+\frac{\fpM\Delta t_{n_1+1}}{2\fpm}},
\end{equation*}
where $\bar{\epsilon}\triangleq\frac{\epsilon\cdot I^2_{n_1+1}(\tau)}{\Delta f_{n_1+1} \Delta f_{n_1+2} \Delta t_{n_1+1}}\cdot \frac{2\fpM}{\fpm}$.
We rearrange such that
\begin{equation}
    \label{eq:der5}
    \begin{aligned}
        \bar{\epsilon}&>\frac{\Delta t_{n_1+2}}{\Delta t_{n_1+1}} + \frac{2 f_{n_1+2}}{\fpm \Delta t_{n_1+1}} - \frac{2f_{n_1+1} \fpM}{{\fpm}^2 \Delta t_{n_1+1}}-\frac{{\fpM}^2}{{\fpm}^2}\\
        &= \frac{\Delta t_{n_1+2}}{\Delta t_{n_1+1}} + 2\frac{f_{n_1+2}\fpm-f_{n_1+1}\fpM}{{\fpm}^2 \Delta t_{n_1+1} }-\frac{{\fpM}^2}{{\fpm}^2}.
    \end{aligned}
\end{equation}

We rewrite \eqref{eq:der5} as
\begin{equation}
\label{eq:der55}
    \bar{\epsilon}>\frac{\Delta t_{n_1+2}}{\Delta t_{n_1+1}}+2\frac{\Delta f_{n_1+1}}{\fpm \Delta t_{n_1+1}}-2f_{n_1+1}\frac{\fpM-\fpm}{{\fpm}^2\Delta t_{n_1+1}}   
    -\frac{{\fpM}^2}{{\fpm}^2}
\end{equation}
We bound the first term on the RHS as $\frac{\Delta t_{n_1+2}}{\Delta t_{n_1+1}}\geq \frac{\Tm}{\TM}$.
For the second term, we have that
\begin{equation}
\label{eq:der555}
2\frac{\Delta f_{n_1+1}}{\fpm\Delta t_{n_1+1}}= \frac{2f'(\bar{\xi}_{n_1+1})}{\fpm}\geq 2,    
\end{equation}
where $\bar{\xi}_{n_1+1}\in\sqb{t_{n_1+1},t_{n_1+2}}$.
Lastly, we bound the third term on the RHS of \eqref{eq:der55} as
\begin{align*}
2f_{n_1+1}\frac{\fpM-\fpm}{{\fpm}^2\Delta t_{n_1+1}}&= 2f_{n_1+1} \frac{f'(\zeta_{\mathsf{M}})-f'(\zeta_{\mathsf{m}})}{{\fpm}^2\Delta t_{n_1+1}}\nonumber
= \frac{2f_{n_1+1}}{{\fpm}^2} \vb{\frac{f'(\zeta_{\mathsf{M}})-f'(\zeta_{\mathsf{m}})}{\zeta_{\mathsf{M}}-\zeta_{\mathsf{m}}}}    \cdot\frac{\vb{\zeta_{\mathsf{M}}-\zeta_{\mathsf{m}}}}{\Delta t_{n_1+1}}\nonumber\\
&= \frac{2f_{n_1+1}}{{\fpm}^2} \cdot \vb{f''(\bar{\zeta}_{n_1+1})}   \cdot\frac{\vb{\zeta_{\mathsf{M}}-\zeta_{\mathsf{m}}}}{\Delta t_{n_1+1}}\label{eq:der55555}
\leq \frac{2f\rb{2\TM+\tau-L}}{{\fpm}^2} \cdot \fppM   \cdot\frac{t_{n_1+3}-t_{n_1+1}}{t_{n_1+2}-t_{n_1+1}}\nonumber
\leq  \frac{2\varphi\rb{t_\delta}}{{\fpm}^2} \cdot \fppM   \cdot\frac{2\TM}{\Tm},\nonumber
\end{align*}
where $t_\delta=-L+2\TM$, $\zeta_{\mathsf{m}},\zeta_{\mathsf{M}}\in\sqb{t_{n_1+1},t_{n_1+3}}$ s.t. $f'(\zeta_{\mathsf{m}})=\fpm$ and $f'(\zeta_{\mathsf{M}})=\fpM$, $\bar{\zeta}_{n_1+1}\in\sqb{t_{n_1+1},t_{n_1+3}}$ s.t. ${f''(\bar{\zeta}_{n_1+1})}=\frac{f'(\zeta_{\mathsf{M}})-f'(\zeta_{\mathsf{m}})}{\zeta_{\mathsf{M}}-\zeta_{\mathsf{m}}}$ and $\fppM=\max_{t\in\sqb{t_{n_1+1},t_{n_1+3}}} \vb{f''(t)}$. Furthermore, the inequalities above also use that $f_{n_1+1}\leq f(\tau+t_\delta)$, which is due to $t_{n_1-1}\leq \tau-L < t_{n_1} < t_{n_1+1}$. Finally,
\begin{equation*}
\fpMb=\max_{t\in\mathbb{S}^\tau_{\delta}} f'(t),\quad
\fpmb=\min_{t\in\mathbb{S}^\tau_{\delta}} f'(t),\quad \fppMb=\max_{t\in\mathbb{S}^\tau_{\delta}} f''(t),
\end{equation*}
where $\mathbb{S}^\tau_{\delta}=\sqb{\tau-L+\Tm/2,\tau-L+4\TM}$.
Using $[t_{n_1+1},t_{{n_1}+3}]\subseteq \mathbb{S}^\tau_{\delta}$, we get that $\fpM\leq\fpMb$, $\fpmb\leq\fpm$, and $\fppM\leq \fppMb$. By plugging these bounds in \eqref{eq:der55}, we get $\bar{\epsilon} > \bar{\epsilon}_{\mathsf{m}}$. According to the definition of $\bar{\epsilon}$ 
\begin{equation}
\label{eq:bound_eps}
    \begin{aligned}
        \epsilon=\frac{\Delta f_{n_1+1}\Delta f_{n_1+2}\Delta t_{n_1+1}}{I_{n_1+1}^2(\tau)}\frac{\fpm}{2\fpM}\bar{\epsilon}\geq \frac{{\fpm}^2\Tm^3}{I_{n_1+1}^2(\tau)}\frac{\fpm}{2\fpM}\bar{\epsilon}\geq \frac{{\fpmb}^3}{2\fpMb}\frac{\Tm^3}{\TM^2}\bar{\epsilon} \geq \frac{{\fpmb}^3}{2\fpMb}\frac{\Tm^3}{\TM^2}\bar{\epsilon}_{\mathsf{m}}.
    \end{aligned}
\end{equation}
\end{proof}
\else
    \begin{proof}
    The proof is in Section \ref{sect:proofs}.
    \end{proof}
\fi

We note that the monotonicity of $\frac{I_{n_1+2}(\tau)}{I_{n_1+1}(\tau)}$ can be guaranteed via \eqref{eq:monotonicity} if $\bar{\epsilon}_{\mathsf{m}}>0$. In this case, we can find $\tau_1$ by solving \eqref{eq:tau1_computation}. The next section tackles the case of multiple pulses.

\subsection{The Case of Multiple Pulses}
The following theorem is our main result, which relaxes the existing assumptions on the filter that enable perfect FRI signal reconstruction from TEM samples.

\begin{theo}[\textbf{FRI Input Recovery}]
\label{theo:1}
Let $g(t)=\sum_{k=1}^K a_k \varphi(t-\tau_k)$ be a FRI input satisfying $\varepsilon_a<\vb{a_k}<g_\infty$, $\Delta\tau_k>\varepsilon_\tau$, $\mathrm{supp}(\varphi)\subseteq[-L,\infty)$, $\norm{\varphi}_\infty=1, \vb{g(t)}\leq g_\infty$. Furthermore, assume that $\varphi\in\mathbb{C}^2\rb{-L,0} \cap\mathbb{C}\rb{\R}$ and $\varphi'(t)>0, t\in\left(-L,0\right)$. 
Let $\cb{t_n}_{n=1}^N$ be the output samples of a TEM with input $g(t)$, such that $4\TM<\varepsilon_\tau$ and $4\TM<L$. Then $\cb{\tau_k,a_k}_{k=1}^K$ can be perfectly recovered from $\cb{t_n}_{n=1}^N$ if
\begin{equation}
    2+\frac{\Tm}{\TM}-\frac{2\varphi\rb{t_\delta}}{{\fpmb}^2}\cdot\fppMb\cdot\frac{2\TM}{\Tm}-\frac{{\fpMb}^2}{{\fpmb}^2}>0,
    \label{eq:suff_cond}
\end{equation}
where $t_\delta=-L+2\TM$, $\fpMb=\max_{t\in\mathbb{S}_{\delta} }\varphi'(t)$, $\fpmb=\min_{t\in\mathbb{S}_{\delta}} \varphi'(t), \fppMb\triangleq \min_{t\in\mathbb{S}_{\delta}} \varphi''(t)$, and $\mathbb{S}_{\delta}= \sqb{-L+\Tm/2,-L+4\TM}$.
\end{theo}
\begin{proof}
    To compute $n_1$ as defined in \eqref{eq:n1}, we note that $\Ln g=0, n<n_1-1$ and $\LN{n_1-1} g\neq 0$. Using \eqref{eq:ttransform}, we get that
    \begin{equation}
    {n}_1=\min_{n\in\cb{1,\dots,N}} \cb{n \setsep \vb{\LN{n-1} g}>0}.
    \end{equation}
    Via the separation property $\varepsilon_\tau>4\TM$, it follows that $t_n\not\in\mathrm{supp} \sqb{\varphi(\cdot-\tau_i)}$, for $n\in\cb{n_1+1,\dots,n_1+3}$ and $i \in \cb{2,\dots,K}$. Therefore pulses $k=2,3\dots,K$ have no effect on $\cb{t_{n_1+1},\dots,t_{n_1+3}}$. Using Lemma \ref{lem:monotonic} via condition $4\TM<L$, \eqref{eq:suff_cond} implies that  $\frac{I_{n_1+2}(\tau)}{I_{n_1+1}(\tau)}$ is strictly increasing, and therefore $\frac{I_{n_1+2}(\tau_1)}{I_{n_1+1}(\tau_1)}=\frac{\mathcal{L}_{n_1+2}g}{\mathcal{L}_{n_1+1}g}$ has a unique solution $\tau_1$. This can be computed via line search to arbitrary accuracy. 

    The amplitude of the first pulse then satisfies ${a}_1=\frac{\mathcal{L}_{n_1+1}g}{I_{n_1+1}({\tau}_1)}$. For the next pulse, we remove the contribution of the first one from the measurements via
\begin{equation*}
    \mathcal{L}_{n+1}^2 g \triangleq \mathcal{L}_{n+1} g - \int_{t_{n+1}}^{t_{n+2}} a_1 \varphi \rb{t-\tau_1} dt,\quad 1\leq n \leq N-1.
\end{equation*}
The process continues recursively for $k=2,\dots,K$, such that
\begin{equation}
\begin{gathered}
\label{eq:nk}
n_k=\min_{n\in\cb{1,\dots,N}} \cb{n \setsep \vb{\mathcal{L}_{n-1}^k g}>0},\\
\frac{I_{n_k+2}(\tau_k)}{I_{n_k+1}(\tau_k)}=\frac{\mathcal{L}_{n_k+2}g}{\mathcal{L}_{n_k+1}g},\quad {a}_k=\frac{\mathcal{L}_{n_k+1}g}{I_{n_k+1}({\tau}_k)}.
\end{gathered}
\end{equation}
\end{proof}

\begin{figure}[!t]	\centerline{\includegraphics[trim={0cm 0cm 0cm 0},clip,width=9cm]{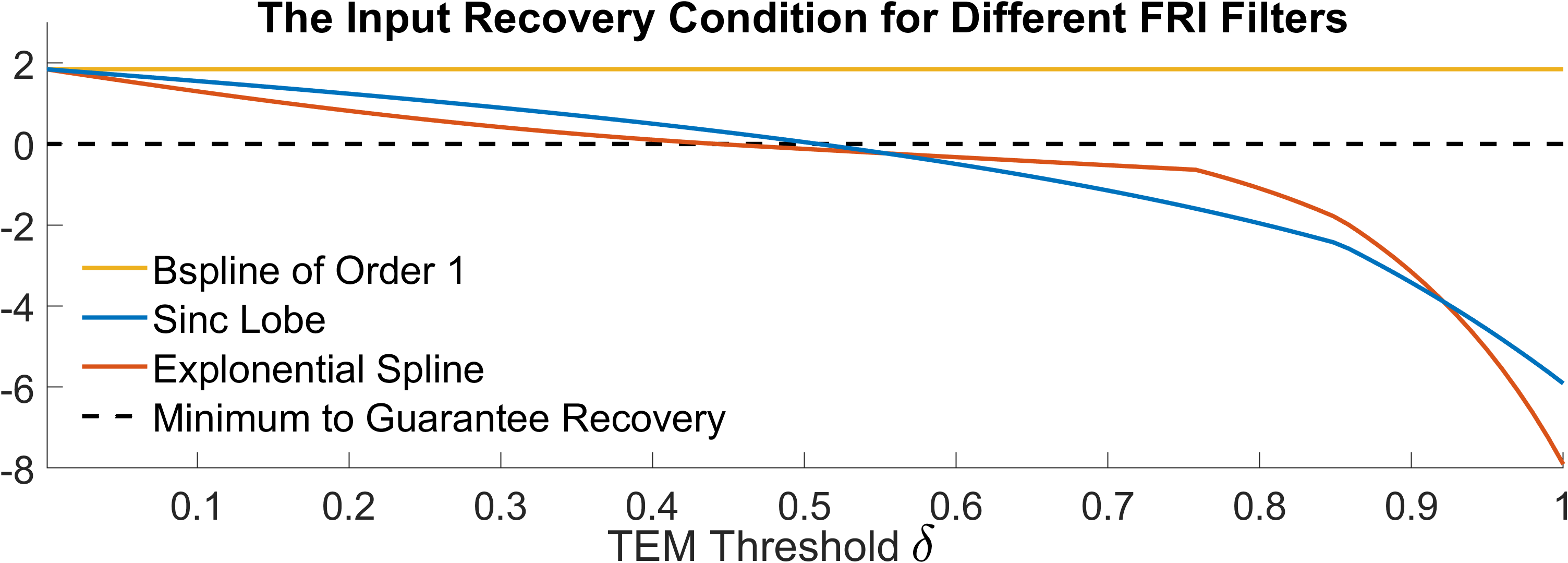}}
	\caption{Evaluating the recovery condition of Theorem \ref{theo:1} for particular cases of filters and a range of $100$ TEM sampling rates. For arbitrarily small $\delta$, all cases converge to the achievable recovery condition of a first order B-spline.}
	\label{fig:rec_condition}
\end{figure} 

First, we note that our separate conditions $4\TM<\varepsilon_\tau$ and $4\TM<L$ imply that $\varepsilon_\tau$ and $L$ are not interrelated as in \cite{Alexandru:2019:J}. In fact, in our case $\varepsilon_\tau$ can be arbitrarily small for high enough sampling rates. Second, we note that \eqref{eq:suff_cond} is both a sampling rate condition as well as a condition on $\varphi(t)$. For example, if $\varphi(t)$ is a first order B-spline, $\fpmb=\fpMb$, $\fppMb=0$ and \eqref{eq:suff_cond} reduces to $\frac{\Tm}{\TM}+1>0$, which is always true. To illustrate how condition \eqref{eq:suff_cond} behaves when changing the sampling rate $\delta$ and filter $\varphi(t)$, we computed the left-hand side of \eqref{eq:suff_cond} for the case of the B-spline of order $1$, the main lobe of a sinc $\varphi(t)=\mathrm{sinc}(\pi t)\cdot\ind_{\sqb{-1,1}}(t)$ and an exponential spline \cite{Alexandru:2019:J}. The results, for $100$ values of $\delta$ uniformly spaced between $10^{-3}$ and $1$, are depicted in \fig{fig:rec_condition}. As is the case for the results in the literature, it turns out that decreasing $\delta$ (increasing the sampling density) is favourable towards input recovery even in this general scenario. In the following we will show rigorously that, under a mild additional assumption, the observations in \fig{fig:rec_condition} hold true in the general case. 

\begin{cor}
\label{cor:1}
Let $g(t)$ be a function satisfying all the conditions in Theorem \ref{theo:1} apart from \eqref{eq:suff_cond}. Under the additional assumption that $\varphi'_+(-L)\neq0$, there exists a TEM threshold $\delta>0$ such that $g(t)$ is perfectly recovered from the corresponding TEM samples $\cb{t_n}_{n=1}^N$.
\end{cor}
\begin{proof}
When $\delta\rightarrow0$, we can simplify equations \eqref{eq:monotonicity} as follows
\begin{equation}
\lim_{\delta\rightarrow0} \frac{{\fpMb}^2}{{\fpmb}^2}=1,\quad \lim_{\delta\rightarrow0} \frac{2\varphi\rb{t_\delta}}{{\fpmb}^2} =0,\quad\lim_{\delta\rightarrow0}\frac{\Tm}{\TM}=\frac{b-g_\infty}{b+g_\infty},
\end{equation}
which holds for both the ASDM and IF TEM. The second limit holds because $\lim_{\delta\rightarrow0}\fpmb=\varphi_+'(-L)\neq0$ and $\varphi(t)=0,t<-L$. By using continuity on $\R$ we get $\varphi(-L)=0$. Thus, $\lim_{\delta\rightarrow0}\bar{\epsilon}_{\mathsf{m}}=1+\frac{\Tm}{\TM}$, which is strictly positive. By writing $\bar{\epsilon}_{\mathsf{m}}(\delta)$ explicitly as a function of $\delta$, it follows that $\exists \delta=\delta_0$ such that $\bar{\epsilon}_{\mathsf{m}}(\delta_0)>0$, leading to  $\rb{\frac{I_{n_1+2}(\tau)}{I_{n_1+1}(\tau)}}'>0$ via  Lemma \ref{lem:monotonic}. This satisfies the conditions for the perfect recovery of $g(t)$ in Theorem \ref{theo:1}. 
\end{proof}

We note that condition $\varphi_+'(-L)=0$ is sufficient, but not necessary. In Section \ref{subsect:bspline} we show that recovery works even when the condition doesn't hold. We note that in practice, due to numerical errors, one would compute $n_k$ in \eqref{eq:nk} via $\vb{\mathcal{L}_{n-1}^k g}>tol$, where $tol$ is a tolerance set by the user. The proposed recovery is summarized in Algorithm \ref{alg:1}, where $\Ln g$ is computed via \eqref{eq:ttransform} for the ASDM and \eqref{eq:ttransformIF} for the IF.
\begin{algorithm}[!t]
\SetAlgoLined
{\bf Data:} $\cb{t_n}_{n=1}^N, \cb{\Ln g}_{n=1}^{N-1}, \varphi(t), tol$.\\
\KwResult{ $\widehat{K}, \ \cb{\widehat{\tau}_k,\widehat{a}_k}_{k=1}^K, \widehat{g}(t)$.}
 \begin{enumerate}
 \itemsep -0pt
\item Compute $\mathcal{L}_n^1g=\Ln g,$ for $n\in\cb{1,\dots,N}$ and $k=1$.
\item While $\exists n\in\cb{1,\dots,N}$ s.t. $\vb{\mathcal{L}_n^k g}>tol$
\begin{enumerate}[leftmargin =*, label = $2\alph*)$]
\item Compute $\widehat{n}_k=\min\limits_{n\in\cb{1,\dots,N}} \cb{n \setsep \vb{\mathcal{L}_{n-1}^k g}>tol} $
\item Compute $I_n(\tau)= \int_{t_n}^{t_{n+1}} \varphi(t-\tau) dt $ for $n\in\cb{\widehat{n}_k+1,\widehat{n}_k+2}$ and $\tau\in\rb{t_{n-2}+L,t_{n-1}+L}$.
\item Find $\widehat{\tau}_k$ from $\tfrac{I_{\widehat{n}_k+2}(\tau)}{I_{\widehat{n}_k+1}(\tau)}=\tfrac{\mathcal{L}_{\widehat{n}_k+2}}{\mathcal{L}_{\widehat{n}_k+1}}$ via line search.
\item Compute $\widehat{a}_k=\frac{\mathcal{L}_{\widehat{n}_k+1}^kg}{I_{\widehat{n}_k+1}(\widehat{\tau}_k)}$.
\item Compute $\mathcal{L}_{n}^{k+1} g \triangleq \mathcal{L}_{n}^k g - \int_{t_n}^{t_{n+1}} \widehat{a}_k \varphi \rb{t-\widehat{\tau}_k} dt$, $k=k+1$.
\end{enumerate}
\item Compute $\widehat{K}=k-1$.
\item Compute $\widehat{g}(t)=\sum_{k=1}^{\widehat{K}} \widehat{a}_k \varphi(t-\widehat{\tau}_k)$.
\end{enumerate}
\caption{Recovery Algorithm.}
\label{alg:1}
\end{algorithm}

\begin{rem}
    The tolerance $tol$ accounts for the effect of noise or numerical inaccuracies, which may lead to $\vb{\LN{n-1}^kg}>0$ even for $n<n_k$. Additionally, we note that Algorithm \ref{alg:1} does not necessarily require $\widehat{n}_k=n_k$ and works with any $\widehat{n}_k\geq n_k$ such that $\tau_k-L<t_{\widehat{n}_k+i}<\tau_k$, $\forall i\in\cb{0,1,2,3}$. 
\end{rem}

\subsection{Sampling Density for the IF TEM with No Bias}
\label{sect:compatibility}

In this subsection we deal with the special case $b=0$ for the IF TEM. Here, the lower bound $\Delta t_n \geq \Tm=\frac{\delta}{b+g_\infty}=\frac{\delta}{g_\infty}$ still holds true, but the upper bound does not, as \eqref{eq:bound_isi} assumes that $g(t)>b-g_\infty>0$, which is no longer true. We note that this bound is only required in our proofs for $\Delta t_{n_k+i}, i\in\cb{0,1,2}$. Assuming that $\varphi(t)$ satisfies  Theorem \ref{theo:1}, the following holds
\begin{equation}
    \begin{aligned}
        \delta=\vb{a_k\int_{t_{n_k+i}}^{t_{n_k+i+1}} \varphi(s-\tau_k) ds}
        \geq \varepsilon_a {\int_{t_{n_k+i}}^{t_{n_k+i+1}} \varphi(s-t_{n_k+i}-L) ds}
        =\varepsilon_a {\int_{-L}^{-L+\Delta t_{n_k+i}} \varphi(s) ds}=\varepsilon_a F\rb{-L+\Delta t_{n_k+i}},
    \end{aligned}
    \label{eq:aux}
\end{equation}
where $F(t)\triangleq \int_{-L}^t \varphi(s) ds$. Above we used that $\varphi(t)>0$ and is increasing for $t\in\sqb{-L,0}$, $t_{n_k+i}>\tau_k-L$. It follows that $F(t)$ is strictly increasing with a strictly increasing inverse $F^{-1}$ for $t\in\sqb{0,L}$. Therefore, \eqref{eq:aux} implies that
\begin{equation}
    \Delta t_{n_k+i} \leq L+ F^{-1}\rb{\frac{\delta}{\varepsilon_a}}\triangleq \TM, 
\end{equation}
where the new definition of $\TM$ only applies when $b=0$, therefore reinforcing the theoretical guarantee in Theorem \ref{theo:1} in this special case. In the next section we analyse the effect of noise on the recovery guarantees.

\section{Robustness to Noise}
\label{sect:noise}
Assume that the noise corrupted input $\widetilde{g}(t)$ satisfies 
\begin{equation}
    \widetilde{g}(t)=\sum_{k=1}^K a_k\varphi(t-\tau_k)+\eta(t), \quad t\in\sqb{0,\tM},
\end{equation}
where $\vb{\eta(t)}\leq\eta_\infty$ is drawn from the uniform distribution on $\sqb{-\eta_\infty,\eta_\infty}$. Function $\widetilde{g}(t)$ is input to a TEM that generates output samples $\cb{{t}_n}_{n=1}^{N}$. 
In this noisy case, we redefine $\wTm\triangleq\frac{2\delta}{b+g_\infty+\eta_\infty}$ $\wTM\triangleq\frac{2\delta}{b-g_\infty-\eta_\infty}$, which satisfies the old notation \eqref{eq:bound_isi} for $\eta_\infty=0$. Given that $\vb{\widetilde{g}(t)}<g_\infty+\eta_\infty$, it is shown similarly to \cite{Lazar:2004:Jc} that
$\wTm\leq \Delta {t}_n \leq \wTM$. The problem proposed is to recover $\cb{\tau_k,a_k}_{k=1}^K$ from the noisy TEM output.

The following theorem derives recovery guarantees in the case of one pulse $\widetilde{g}(t)=a_1\varphi(t-\tau_1)+\eta(t)$. This corresponds to one iteration of Step 2) in Algorithm \ref{alg:1}.
\begin{theo}[\textbf{Noisy Input Recovery}]
\label{theo:2}
    Let $\cb{\widetilde{t}_n}_{n=1}^{N}$ be the TEM output in response to $\widetilde{g}(t)=a_1\varphi(t-\tau_1)+\eta(t)$, where $\vb{\eta(t)}<\eta_\infty$, such that $\etab\triangleq\eta_{\infty} \TM <\fpmb \cdot \varepsilon_a \wTm$. Furthermore, let $\bar{\epsilon}_{\mathsf{m}}$ defined as in \eqref{eq:monotonicity} and assume that \eqref{eq:suff_cond} is true. Then $\widehat{\tau}_1$ computed via Algorithm \ref{alg:1} satisfies
    \begin{equation}
        \vb{\tau_1-\widehat{\tau}_1}\leq e_\tau\triangleq\frac{2g_\infty \etab \bar{\epsilon}_{\mathsf{m}}}{\rb{\varepsilon_a \wTm\fpmb-\etab}^2}\frac{2\fpMb}{{\fpmb}^3}\frac{\TM^3}{\Tm^3}.
    \end{equation}    
    Assuming $e_{\tau}<\Tm/2$, then $\vb{a_1-\widehat{a}_1}\leq e_a\triangleq\frac{\TM}{\Tm^2}\frac{2 e_\tau g_\infty+\etab}{{\fpmb}^2}$.
\end{theo}
\if\NoisyRecInAppendix\FL
\begin{proof}
From Theorem \ref{theo:1} we have that, if \eqref{eq:suff_cond} is true, then $\frac{I_{n_1+2}(\tau)}{I_{n_1+1}(\tau)}$ is strictly increasing and $\tau_1$ is correctly computed by solving $\frac{I_{n_1+2}(\tau)}{I_{n_1+1}(\tau)}=\frac{\mathcal{L}_{n_1+2}{g}}{\mathcal{L}_{n_1+1}{g}}$ for $\tau$. There are two factors contributing to the error $\vb{\tau_1-\widehat{\tau}_1}$: the measurement error on $\frac{\mathcal{L}_{n_1+2}\widetilde{g}}{\mathcal{L}_{n_1+1}\widetilde{g}}$, and the slope of $\frac{I_{n_1+2}(\tau)}{I_{n_1+1}(\tau)}$, which will be evaluated as follows. We start by providing bounds for $\mathcal{L}_{n_1+1}g$. We consider two cases

$\mathbf{1)\ a_1>0}$. For $t\in\sqb{t_{n_1+1},t_{n_1+2}}$ we have that $\varphi(t-\tau_1)>0$, $\varphi'(t-\tau_1)>\fpmb>0$ and $g(t)>0$, $g'(t)>\varepsilon_a\fpmb>0$. Therefore, we can derive the following bounds
\begin{equation}
\label{eq:bounds_Lg}
\begin{gathered}
\varepsilon_a \fpmb (t-t_{n_1+1})\leq g(t)\leq g_\infty, \quad t\in\sqb{t_{n_1+1},t_{n_1+2}},\\
\varepsilon_a \wTm\fpmb\leq\mathcal{L}_{n_1+1}g\leq g_\infty\wTM.
\end{gathered}
\end{equation}
We know that $\etab<\varepsilon_a \wTm\fpmb$, where $\etab\triangleq\eta_\infty \wTM$. Furthermore, $\mathcal{L}_{n_1+1}\eta<\etab$, $\mathcal{L}_{n_1+2}\eta<\etab$, and therefore
\begin{equation}
    \label{eq:ratio_bound}
    \frac{\mathcal{L}_{n_1+2}\widetilde{g}}{\mathcal{L}_{n_1+1}\widetilde{g}}\leq\frac{\mathcal{L}_{n_1+2}{g}+\etab}{\mathcal{L}_{n_1+1}{g}-\etab}.
\end{equation}
The measurement error then can be bounded as
\begin{equation}
\label{eq:abs_bound1}
\begin{aligned}
    \frac{\mathcal{L}_{n_1+2}\widetilde{g}}{\mathcal{L}_{n_1+1}\widetilde{g}}-\frac{\mathcal{L}_{n_1+2}{g}}{\mathcal{L}_{n_1+1}{g}}\leq    \frac{\mathcal{L}_{n_1+2}{g}+\etab}{\mathcal{L}_{n_1+1}{g}-\etab}-\frac{\mathcal{L}_{n_1+2}{g}}{\mathcal{L}_{n_1+1}{g}}
    \leq \etab\cdot\frac{\mathcal{L}_{n_1+1}{g}+\mathcal{L}_{n_1+2}{g}}{\mathcal{L}_{n_1+1}{g}\rb{\mathcal{L}_{n_1+1}{g}-\etab}}
    \leq \frac{2g_\infty\wTM \etab}{\varepsilon_a \wTm\fpmb\rb{\varepsilon_a \wTm\fpmb-\etab}}.
\end{aligned}
\end{equation}
Similarly, it can be shown that
\begin{equation}
\label{eq:abs_bound2}
    \frac{\mathcal{L}_{n_1+2}{g}}{\mathcal{L}_{n_1+1}{g}}-\frac{\mathcal{L}_{n_1+2}\widetilde{g}}{\mathcal{L}_{n_1+1}\widetilde{g}}    \leq \frac{2g_\infty\wTM \etab}{\varepsilon_a \wTm\fpmb\rb{\varepsilon_a \wTm\fpmb+\etab}}.
\end{equation}
Using \eqref{eq:abs_bound1} and \eqref{eq:abs_bound2}, we get that 
\begin{equation}
    \label{eq:bound_measurements}
    \vb{\frac{\mathcal{L}_{n_1+2}\widetilde{g}}{\mathcal{L}_{n_1+1}\widetilde{g}}-\frac{\mathcal{L}_{n_1+2}{g}}{\mathcal{L}_{n_1+1}{g}}}\leq \frac{2g_\infty\wTM \etab}{\rb{\varepsilon_a \wTm\fpmb-\etab}^2}.
\end{equation}

$\mathbf{2)\ a_1<0}$. We repeat derivations \eqref{eq:bounds_Lg}-\eqref{eq:bound_measurements} where $g(t)$ is replaced by $-g(t)$, yielding the same bound \eqref{eq:bound_measurements}.

Lastly, the error for computing $\tau$ via $\frac{I_{n_1+2}(\tau)}{I_{n_1+1}(\tau)}=\frac{\mathcal{L}_{n_1+2}\widetilde{g}}{\mathcal{L}_{n_1+1}\widetilde{g}}$ is
\begin{equation*}
\begin{aligned}
    \vb{\tau_1-\widehat{\tau}_1}&< \frac{2g_\infty\wTM \etab}{\rb{\varepsilon_a \wTm\fpmb-\etab}^2} \cdot \sqb{\min_{\tau} \rb{\frac{I_{n_1+2}(\tau)}{I_{n_1+1}(\tau)}}'}^{-1}\leq \frac{2g_\infty\wTM \etab \bar{\epsilon}_{\mathsf{m}}}{\rb{\varepsilon_a \wTm\fpmb-\etab}^2}\cdot\frac{2\fpMb}{{\fpmb}^3}\frac{\TM^2}{\Tm^3}=e_\tau,
\end{aligned}
\end{equation*}
where the last inequality uses \eqref{eq:bound_eps}. For $a_1$, we note that
\begin{equation}
    a_1=
    \frac{\mathcal{L}_{n_1+1}g}{\mathcal{L}_{n_1+1}\sqb{\varphi(\cdot-\tau_1)}},\quad \widehat{a}_1=
    \frac{\mathcal{L}_{n_1+1}g+\mathcal{L}_{n_1+1}\eta}{\mathcal{L}_{n_1+1}\sqb{\varphi(\cdot-\widehat{\tau}_1)}},
\end{equation}
and thus
\begin{align*}
    \vb{a_1-\widehat{a}_1}&\leq\frac{\vb{\mathcal{L}_{n_1+1}g\cdot{\mathcal{L}_{n_1+1}\sqb{\varphi(\cdot-\widehat{\tau}_1)-\varphi(\cdot-{\tau}_1)}}}}{\vb{\mathcal{L}_{n_1+1}\sqb{\varphi(\cdot-\tau_1)}}\cdot\vb{\mathcal{L}_{n_1+1}\sqb{\varphi(\cdot-\widehat{\tau}_1)}}}+\mathcal{N},
\end{align*}
where $\mathcal{N}\triangleq\frac{\vb{\mathcal{L}_{n_1+1}\eta\cdot\mathcal{L}_{n_1+1}\sqb{\varphi(\cdot-\tau_1)}}}{\vb{\mathcal{L}_{n_1+1}\sqb{\varphi(\cdot-\tau_1)}}\cdot\vb{\mathcal{L}_{n_1+1}\sqb{\varphi(\cdot-\widehat{\tau}_1)}}}$.
We use that 
\begin{equation}
    \mathcal{L}_{n_1+1}\sqb{\varphi(\cdot-\widehat{\tau}_1)}=\int_{t_{n_1+1}+\tau_1-\widehat{\tau}_1}^{t_{n_1+2}+\tau_1-\widehat{\tau}_1} f(s) ds,
    \label{eq:filter_tau_hat}
\end{equation}
\begin{align*}
    \vb{{a}_1-\widehat{a}_1}&\leq\frac{\vb{\mathcal{L}_{n_1+1}g}\cdot\vb{\int_{\mathbb{M}}f(s)ds}}{\vb{\mathcal{L}_{n_1+1}\sqb{\varphi(\cdot-\tau_1)}}\cdot\vb{\mathcal{L}_{n_1+1}\sqb{\varphi(\cdot-\widehat{\tau}_1)}}}+\mathcal{N},
\end{align*}
s.t. $\mathbb{M}\triangleq \sqb{t_{n_1+1},t_{n_1+1}+\widehat{\tau}_1-\tau_1}\cup \sqb{t_{n_1+2}+\widehat{\tau}_1-\tau_1,t_{n_1+2}}$. Using $\vb{\mathcal{L}_{n_1+1}g}<g_\infty\wTM$, $\vb{\mathcal{L}_{n_1+1}\eta}<\etab$, $\vb{\int_{\mathbb{M}}f(s)ds}\leq 2 e_\tau$, and $\vb{\mathcal{L}_{n_1+1}\varphi(\cdot-\tau_1)}<\wTM$,
\begin{equation}
\label{eq:amp_err_bound}
\begin{aligned}
    \vb{{a}_1-\widehat{a}_1}&\leq \frac{\wTM(2 e_\tau g_\infty+\etab)}{\vb{\mathcal{L}_{n_1+1}\sqb{\varphi(\cdot-\tau_1)}}\cdot\vb{\mathcal{L}_{n_1+1}\sqb{\varphi(\cdot-\widehat{\tau}_1)}}}.
\end{aligned}
\end{equation}
The following holds from \eqref{eq:filter_tau_hat}, using that $e_{\tau}<\Tm/2$.
\begin{equation}
\begin{gathered}
\mathcal{L}_{n_1+1}\sqb{\varphi(\cdot-\widehat{\tau}_1)}\geq \int_{-L+\tau+\Tm/2}^{-L+\tau+3\Tm/2} f(s) ds \geq \Tm \fpmb.
\end{gathered}
\end{equation}
Using \eqref{eq:bounds_Lg} we get $\mathcal{L}_{n_1+1}\sqb{\varphi(\cdot-{\tau}_1)}\geq \Tm \fpmb$, which completes the proof via \eqref{eq:amp_err_bound}.

\end{proof}
\else
    \begin{proof}
    The proof is in Section \ref{sect:proofs}.
    \end{proof}
\fi

We extend the result in Theorem \ref{theo:2} recursively for $K$ pulses. Specifically, we assume that Step 2) of Algorithm \ref{alg:1} was computed $K-1$ times, for $K\geq2$, and that $\cb{e_{\tau_k},e_{a_k}}_{k=1}^{K-1}$ are known, and we derive $e_{\tau_K}$ and $e_{a_K}$. In a noiseless scenario, $a_{K}$ and $\tau_K$ could be perfectly recovered from local integrals $\LN{n} a_K\varphi(\cdot-\tau_K)$. Thus, we first estimate a noise bound for the local integral of the $K$th pulse $\eta_{\delta,K}$ satisfying $\vb{\LN{n} \sqb{a_K\varphi(\cdot-\tau_K)}-\LN{n}^{K}\widetilde{g}}<\eta_{\delta,K}$. Then $e_{\tau_K}$ and $e_{a_K}$ can be derived by via (\ref{eq:ratio_bound}-\ref{eq:amp_err_bound}), where $\etab$ is substituted with $\eta_{\delta,K}$. Using Algorithm \ref{alg:1} step 2e),
\begin{equation*}
\begin{aligned}    \LN{n}^{K}\widetilde{g}=\rb{\LN{n} g+\LN{n}{\eta}}-\sum_{k=1}^{K-1}\LN{n}\sqb{\widehat{a}_k\varphi(\cdot-\widehat{\tau}_k)}.
\end{aligned}
\end{equation*} 
When computing $\vb{\LN{n} a_K\varphi(\cdot-\tau_k)-\LN{n}^{K}\widetilde{g}}$ we get
\begin{equation}
    \begin{gathered}
    \vb{\sum_{k=1}^{K-1}{\LN{n} \sqb{{a}_k\varphi(\cdot-{\tau}_k)-\widehat{a}_k\varphi(\cdot-\widehat{\tau}_k)}+\LN{n} {\eta}}}
    \leq\sum_{k=1}^{K-1}{\LN{n} \vb{{a}_k\varphi(\cdot-{\tau}_k)-\widehat{a}_k\varphi(\cdot-\widehat{\tau}_k)}+\etab}
    \end{gathered}
    \label{eq:n}
\end{equation}
We bound the absolute value in \eqref{eq:n} as
\begin{equation}
    \begin{aligned}
    \vb{{a}_k\varphi(t-{\tau}_k)-\widehat{a}_k\varphi(t-\widehat{\tau}_k)}
    &=\vb{{a}_k\sqb{\varphi(t-{\tau}_k)-\varphi(t-\widehat{\tau}_k)}+\rb{a_k-\widehat{a}_k}\varphi(t-\widehat{\tau}_k)}\\
    &\leq \vb{{a}_k} \varphi'_\mathsf{M}  e_{\tau_k}+e_{a_k}\leq g_\infty \varphi'_\mathsf{M}  e_{\tau_k}+e_{a_k},
    \end{aligned}
    \label{eq:etau_ea}
\end{equation}
where $\varphi'_\mathsf{M}\triangleq \max \cb{\norm{\varphi_-'}_{\infty},\norm{\varphi_+'}_{\infty}}$ and $\varphi_-'(t), \varphi_+'(t)$ are the left and right derivatives, respectively. Using \eqref{eq:etau_ea} in \eqref{eq:n},
\begin{equation}
\begin{gathered}
    \vb{\LN{n} a_K\varphi(\cdot-\tau_k)-\LN{n}^{K}\widetilde{g}}\leq \eta_{\delta,K},\quad 
    \eta_{\delta,K}\triangleq\wTM\sum_{k=1}^{K-1}\rb{g_\infty \varphi'_\mathsf{M} e_{\tau_k}+e_{a_k}}+\etab.
\end{gathered}
\label{eq:etadelK}
\end{equation}

By repeating steps (\ref{eq:ratio_bound}-\ref{eq:amp_err_bound}) for $\eta_{\delta,K}$, we get that 
    \begin{equation}
    \begin{gathered}
        e_{\tau_K}=\frac{2g_\infty \etab \bar{\epsilon}_{\mathsf{m}}}{\rb{\varepsilon_a \wTm\fpmb-\eta_{\delta,K}}^2}\frac{2\fpMb}{{\fpmb}^3}\frac{\TM^3}{\Tm^3}.
    \end{gathered}
    \label{eq:errors}
    \end{equation}   
Furthermore, assuming $e_{\tau_K}<\Tm/2$, $e_{a_K}=\frac{\TM}{\Tm^2}\frac{2 e_\tau g_\infty+\eta_{\delta,K}}{{\fpmb}^2}$.
Then the estimation errors $\cb{e_{\tau_k}, e_{a_k}}_{k=1}^K$ for Algorithm \ref{alg:1} can be computed recursively as above via  $\eta_{\delta,1}=\etab$.

\vspace{2em}
\section{Numerical and Hardware Experiments}
\label{sect:numerical_study}
Here we test our recovery approach for a wide selection of FRI filters including filters previously used in the literature and new synthetically generated filters. We consider the case of a simulated ASDM and also an analog hardware implementation. Furthermore, to allow a comparison with the existing methods, we show examples using sampling setups proposed in the literature \cite{Alexandru:2019:J,Hilton:2023:C}. We evaluate the recovery error of the FRI parameters using $\mathsf{Err}_{\tau}$ and $\mathsf{Err}_{a}$, defined as
\begin{equation}
    \mathsf{Err}_{\tau}=\tfrac{1}{K}\sum_{k=1}^K 100\cdot\frac{\vb{\widehat{\tau}_k-\tau_k}}{\vb{\tau_k}},\quad \mathsf{Err}_{a}=\tfrac{1}{K}\sum_{k=1}^K 100\cdot\frac{\vb{\widehat{a}_k-a_k}}{\vb{a_k}}.
\end{equation}
Furthermore, we evaluate the recovery error for $g(t)$ as $\mathsf{Err}_g=100\cdot\frac{\norm{g-\widehat{g}}_2}{\norm{g}_2}\ (\%)$.
This section is organised as follows. Sections \ref{subsect:bspline} and \ref{subsect:random} present examples with a B-spline filter and a randomly generated filter, respectively. Section \ref{sect:comparative} shows recovery examples with an E-spline filter and a hyperbolic secant filter. Finally, Section \ref{sect:hardware} presents a recovery example for a hardware implementation of an ASDM.

\subsection{FRI Input Recovery with a B-spline Filter}
\label{subsect:bspline}
We first evaluate Algorithm \ref{alg:1} for a filter $\varphi_b$ representing a B-Spline of order $3$ scaled such that it is supported in $[-1,1]$ and has amplitude $1$. We note that $\varphi_b'(-L)=0$, and therefore the condition in Corollary \ref{cor:1} is not true. Even so, we demonstrate numerically that recovery works in this case. Signal $g(t)$ was generated for $\tau_1=1.15, \tau_2=2.52, \tau_3=4.74, \tau_4=5.81, a_1=3.22, a_2=-2.34, a_3=2.87, a_4=3.54$. Signal $g(t)$ was sampled with an ASDM with parameters $b=12,\delta=1, g_0=0$. The input Diracs, signal $g(t)$, TEM output along with the reconstructed signals via Algorithm \ref{alg:1} with $tol=0.05$ are depicted in \fig{fig:spline_random}(a). The resulting errors are $\mathsf{Err}_{\tau}=0.098\%$, $\mathsf{Err}_{a}=0.39\%$, and $\mathsf{Err}_g=0.86\%$.

\subsection{FRI Recovery with a Random Filter}
\label{subsect:random}
To demonstrate the generalization enabled by the proposed algorithm, we considered the case of a randomly generated filter $\varphi_r(t)$, which was not validated numerically or demonstrated theoretically in the existing literature. The random filter consists of a random increasing function followed by a random decreasing function. To generate the first function, we convolved a uniform random noise sequence with a B-spline of degree $9$. We subtracted the minimum to make it strictly positive, then integrated it and scaled it to be in the interval $[0,1]$. The second function was generated in the same way, only here we subtracted the maximum to make the final function strictly decreasing. The resulted filter $\varphi_r(t)$ was used to generate an input $g(t)$ for $\tau_1=1.21, \tau_2=2.26, \tau_3=4.65, a_1=3.22, a_2=-2.33, a_3=-2.87$. Signal $g(t)$ was input to an ASDM with parameters $\delta=0.5, b=12, g_0=0$. The input Diracs, TEM input, TEM output, the recovery of the input Diracs and of the TEM input via Algorithm \ref{alg:1} with $tol=0.006$ are depicted in \fig{fig:spline_random}(b). We note that the sampling rate is higher compared to \fig{fig:spline_random}(a), mainly due to using a filter with an irregular shape. However, as shown in Corollary \ref{cor:1} recovery is still possible for $\delta$ small enough. The corresponding recovery errors are $\mathsf{Err}_{\tau}=0.12\%$, $\mathsf{Err}_{a}=1.33\%$, and $\mathsf{Err}_g=1.29\%$.

\begin{figure}[!t]
\begin{center}
\includegraphics[trim={0cm 0cm 0cm 0},clip,width=0.47\textwidth]{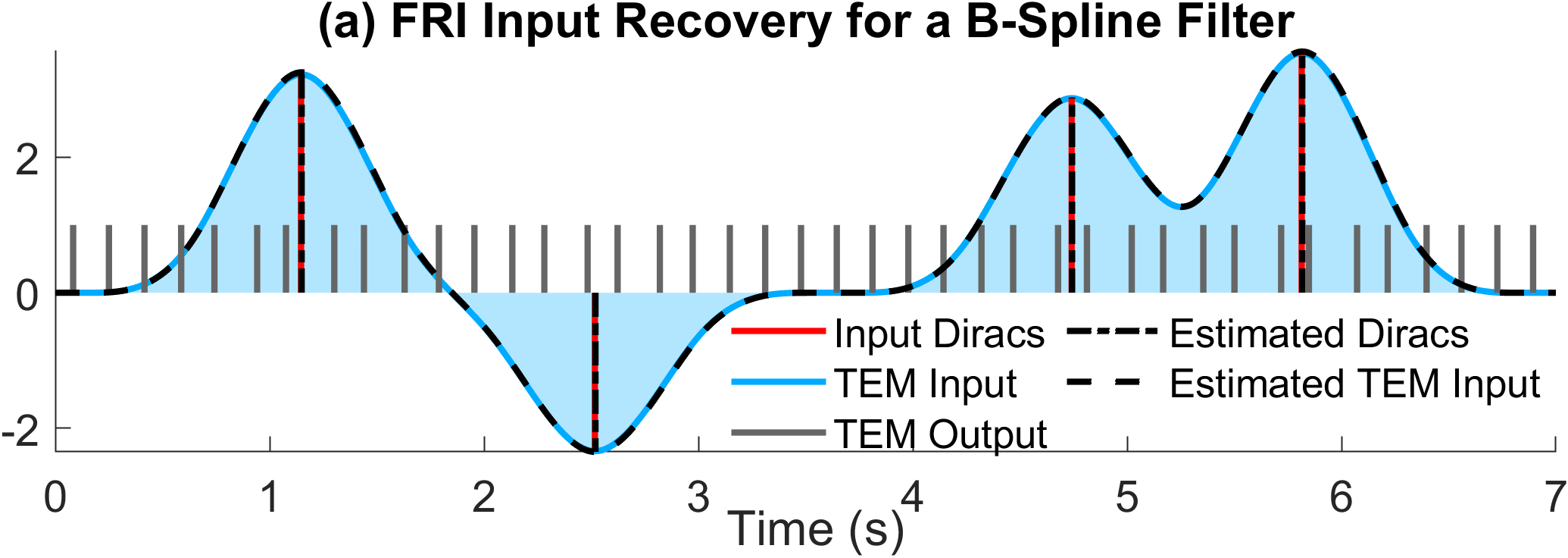}
\includegraphics[trim={0cm 0cm 0cm 0},clip,width=0.47\textwidth]{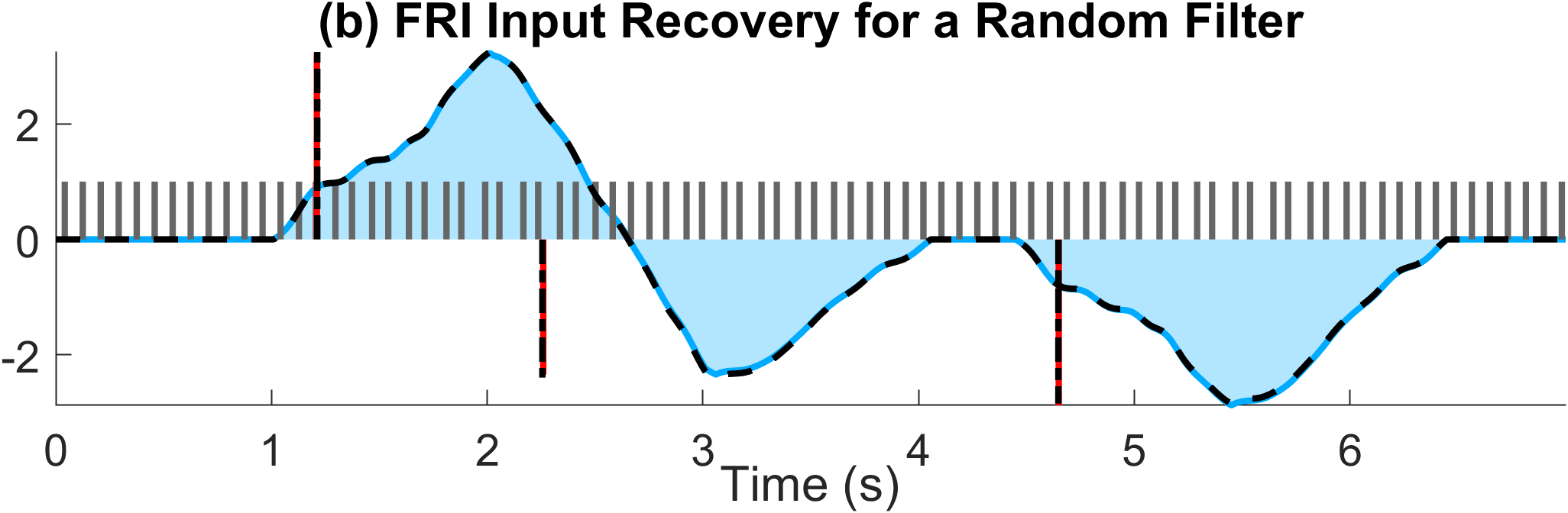}
\end{center}
\caption{Evaluating Algorithm \ref{alg:1} with a B-spline filter and a randomly generated filter. This result demonstrates the applicability of the proposed method for a wider class of filters than previously possible.}
  \label{fig:spline_random}
\end{figure}

\subsection{Comparison with Existing Methods}
\label{sect:comparative}
Here we compare Algorithm \ref{alg:1} with the method in \cite{Alexandru:2019:J}, for the case of an E-spline filter, and also with the method in \cite{Hilton:2023:C} for the case of a squared hyperbolic secant filter. We implement Algorithm \ref{alg:1} based on the IF TEM model to allow a comparison with the results in the literature. As the method in \cite{Alexandru:2019:J} is restricted to specific class of filters, we first selected a second order E-spline for both methods, defined as 
\[
    \varphi_e(t)=\left\lbrace 
    \begin{array}{cc}
         \frac{e^{\alpha_1-\alpha_0}}{\alpha_1-\alpha_0}e^{-\alpha_0 t}+\frac{e^{-\alpha_1+\alpha_0}}{\alpha_0-\alpha_1}e^{-\alpha_1 t},& -L\leq t \leq -L/2 \\
        \frac{1}{\alpha_0-\alpha_1}e^{-\alpha_0 t} + \frac{1}{\alpha_1-\alpha_0}e^{-\alpha_1 t}, & -L/2\leq t \leq 0,\\
        0, & \mathrm{otherwise},
    \end{array} 
    \right.
\]
where $\alpha_0=0-\jmath 1.047, \alpha_1=0+\jmath 1.047$. The TEM input is
\begin{equation}
    g(t)=\sum_{k=1}^4 a_k \varphi_e(t-\tau_k)+\eta(t),
\end{equation}
where $\tau_1=1, \tau_2=4, \tau_3=7, \tau_4=9.5$, $a_1=4.36, a_2=4.04, a_3=4.92,$ and $ a_4=5.45$. Furthermore, $\eta(t)$ is uniform noise bounded by $\vb{\eta(t)}\leq\eta_\infty=0.05$.

The output of the filter is encoded with an IF model with parameters $\delta=0.5, b=0$. The E-spline, IF encoding and the Prony based recovery in \cite{Alexandru:2019:J} were implemented using publicly available software \cite{Alexandru:2019:W}. The signal $g(t)$, the IF output samples $\cb{t_n}$ and reconstructions with Algorithm \ref{alg:1} using $tol=0.05$ and the Prony based method in \cite{Alexandru:2019:J} are depicted in \fig{fig:comparison}(a1). 
We computed $\mathsf{Err}_g$ and averaged it over $100$ different noise signals $\eta(t)$. This resulted in $0.47\%$ for the Prony based recovery and $0.43\%$ for the proposed method.
The Prony based recovery is not guaranteed to work for overlapping pulses. We adjust the pulses to new time locations $\tau_k^\ast=0.4\cdot \tau_k$. The results, depicted in \fig{fig:comparison}(a2), show that the proposed method is able to handle a significant amount of overlapping. This is primarily due to step 2e) in Algorithm \ref{alg:1}, which removes the contribution of each identified pulse to future TEM samples.

\begin{figure*}[!t]
\begin{center}
\includegraphics[trim={0cm 0cm 0cm 0},clip,width=0.45\textwidth]{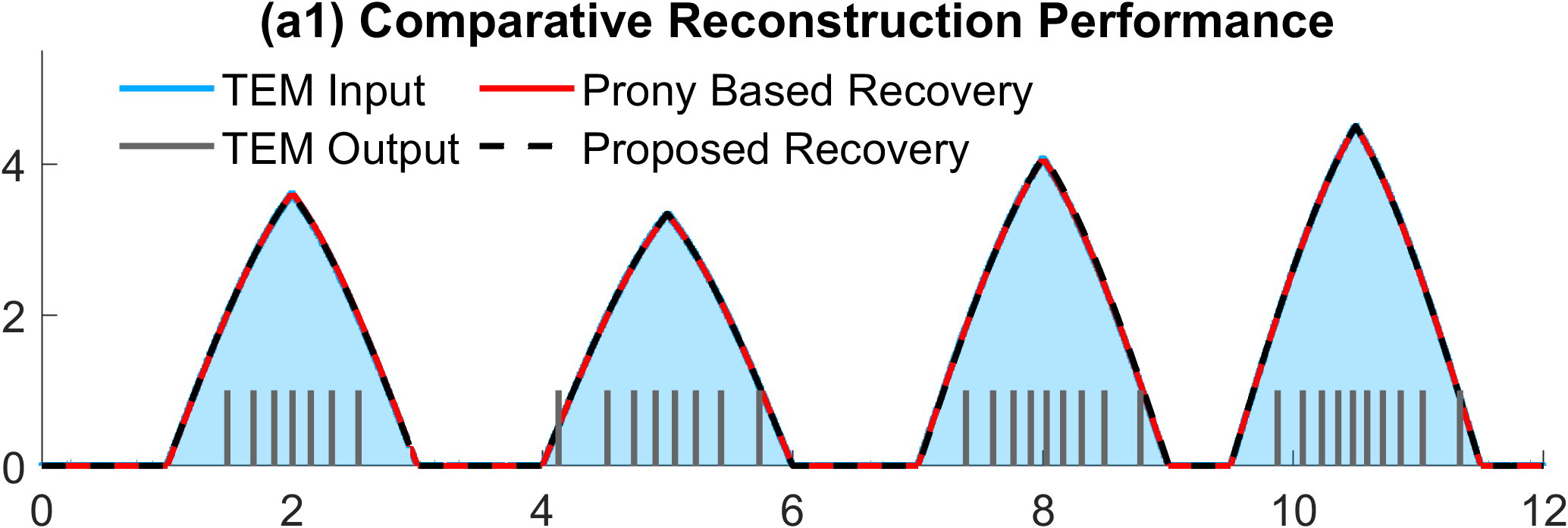}
\includegraphics[trim={0cm 0cm 0cm 0},clip,width=0.45\textwidth]{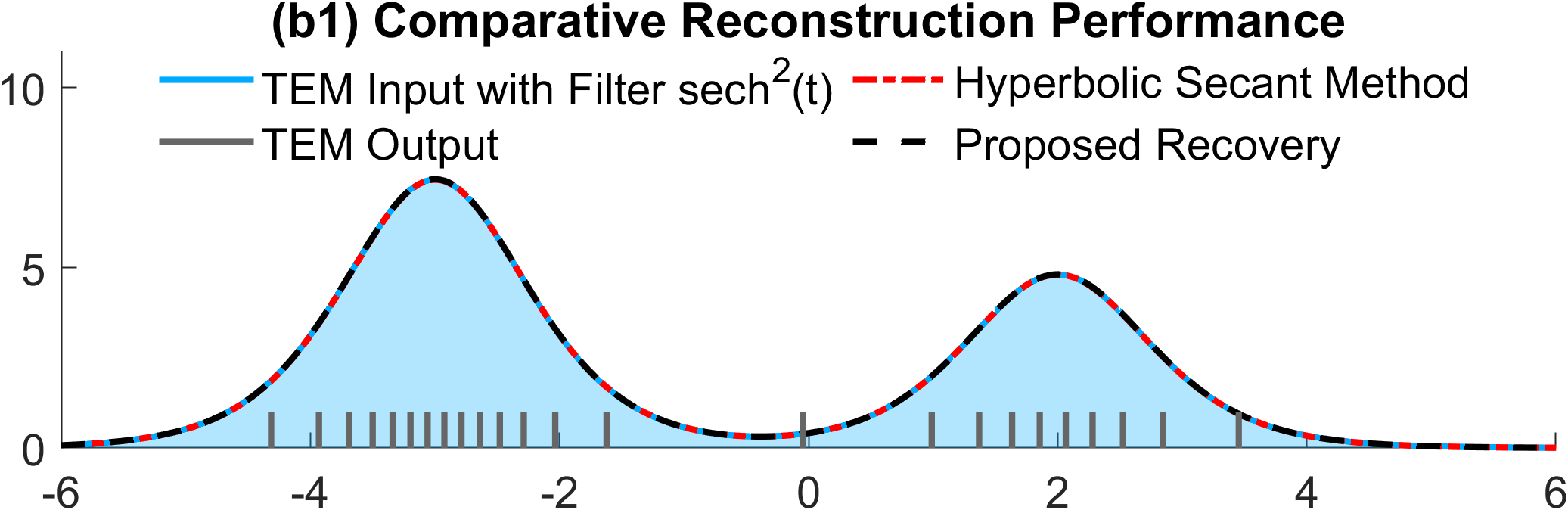}\\
\includegraphics[trim={0cm 0cm 0cm 0},clip,width=0.45\textwidth]{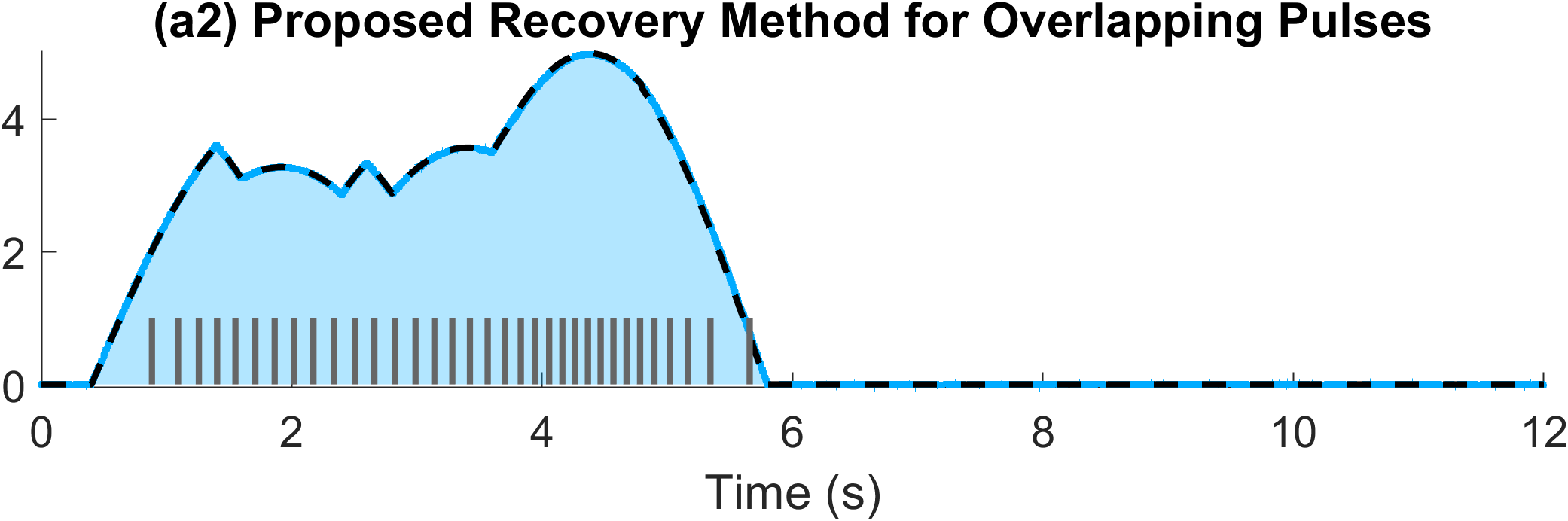}
\includegraphics[trim={0cm 0cm 0cm 0},clip,width=0.45\textwidth]{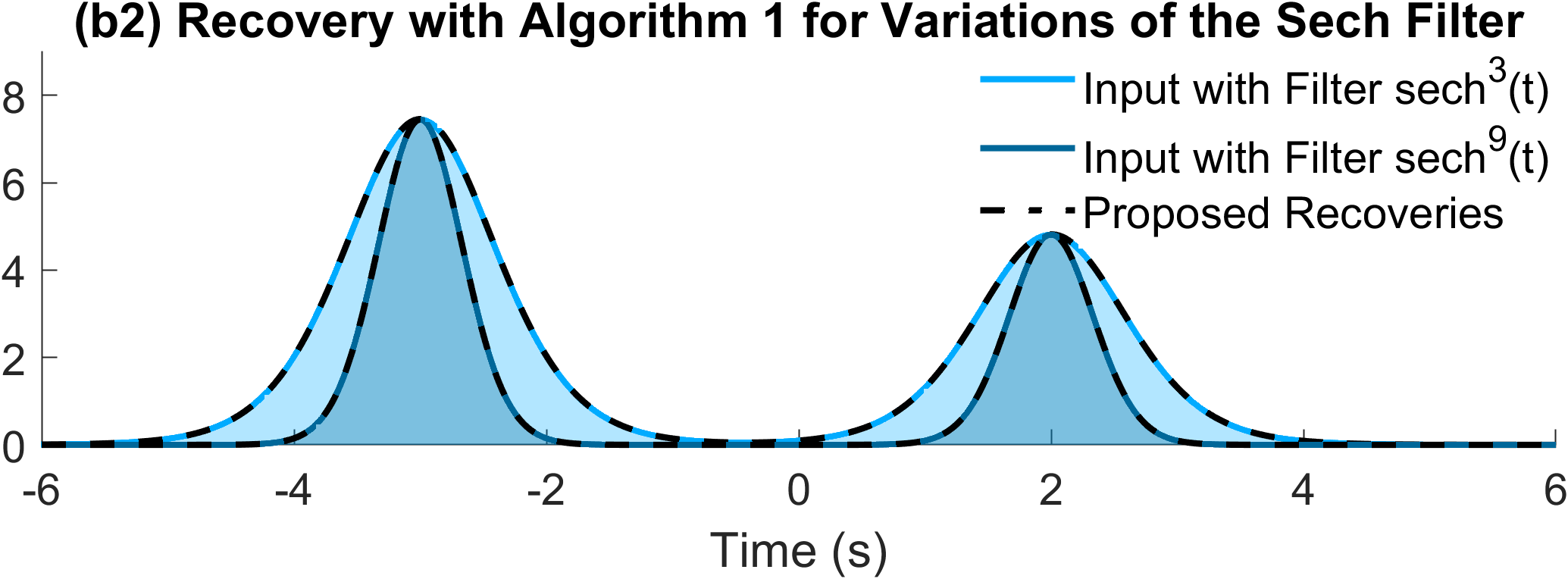}
\end{center}
\caption{ (a) Comparative recovery with the proposed method and the Prony-based method in \cite{Alexandru:2019:J}. (a1) The noisy FRI input is based on an E-spline of order $1$. With no overlaps, both methods recover the TEM input correctly. (a2) The pulses overlap, and thus the conditions in \cite{Alexandru:2019:J} are not satisfied and the recovery is unstable. (b) Comparative recovery with the method in \cite{Hilton:2023:C}. (b1) The input generated via $\varphi(t)=\mathrm{sech}^2(t)$ is recovered with Algorithm \ref{alg:1} and the method in \cite{Hilton:2023:C}. (b2) The filters are $\varphi(t)=\mathrm{sech}^3(t)$ and $\varphi(t)=\mathrm{sech}^9(t)$, which don't satisfy the conditions in \cite{Hilton:2023:C}, leading to unstable reconstructions. 
}
  \label{fig:comparison}
\end{figure*}
We further compared the proposed method with the method in \cite{Hilton:2023:C}, which is based on the assumption that the filter is constructed with the hyperbolic secant function defined as $\mathrm{sech}(t)=\frac{2e^t}{1+2e^t}$. Here we consider the case where $\varphi(t)=\mathrm{sech}^2(t)$, as presented in \cite{Hilton:2023:C}. We generated the TEM input as $g(t)=\sum_{k=1}^2 a_k\mathrm{sech}^2(t-\tau_k)$ where $\tau_1=-3,\tau_2=2, a_1=7.44, a_2=4.8$. The TEM used is an IF model with $\delta=1, b=0$. The signal $g(t)$, the IF output samples $\cb{t_n}$ and reconstructions with Algorithm \ref{alg:1} with $tol=0.01$ and the method in \cite{Hilton:2023:C} are depicted in \fig{fig:comparison}(b1). The resulted errors are $\mathsf{Err}_{\tau}=0.03\%$, $\mathsf{Err}_{a}=0.005\%$, and $\mathsf{Err}_g=0.13\%$ for the method in \cite{Hilton:2023:C} and $\mathsf{Err}_{\tau}=0.04\%$, $\mathsf{Err}_{a}=0.14\%$, and $\mathsf{Err}_g=0.17\%$ for Algorithm \ref{alg:1}. To exploit the flexibility of the proposed method, we repeated the experiment with the same parameters by changing the filter to $\varphi(t)=\mathrm{sech}^3(t)$ and $\varphi(t)=\mathrm{sech}^9(t)$. The method in \cite{Hilton:2023:C} is not compatible with
these filters and leads to unstable reconstructions. We note that these filters also don't satisfy the conditions of Theorem \ref{theo:1}, as they are supported on the real axis. Interestingly, Algorithm \ref{alg:1} still performs well with errors $\mathsf{Err}_g=0.2\%$ and $\mathsf{Err}_g=0.12\%$, respectively. The results are illustrated in \fig{fig:comparison}(b2).

\subsection{Hardware Experiment}
\label{sect:hardware}
We validate the proposed recovery method using a hardware implementation of the acquisition pipeline in \fig{fig:asdm} as follows. The FRI input signal was generated on a PC and fed to the circuit via the audio channel. The input was subsequently amplified prior to being injected into the ASDM hardware, depicted in \fig{fig:hw}(a). 

\begin{figure*}[!t]
\includegraphics[trim={0cm 0cm 0cm 0},clip,width=0.31\textwidth]{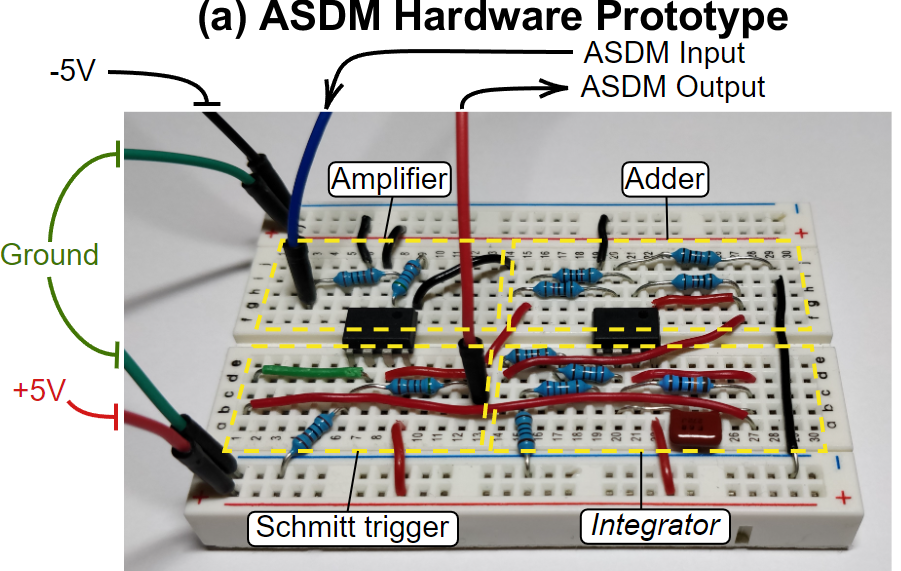}
\includegraphics[trim={0cm 0cm 0cm 0},clip,width=0.32\textwidth]{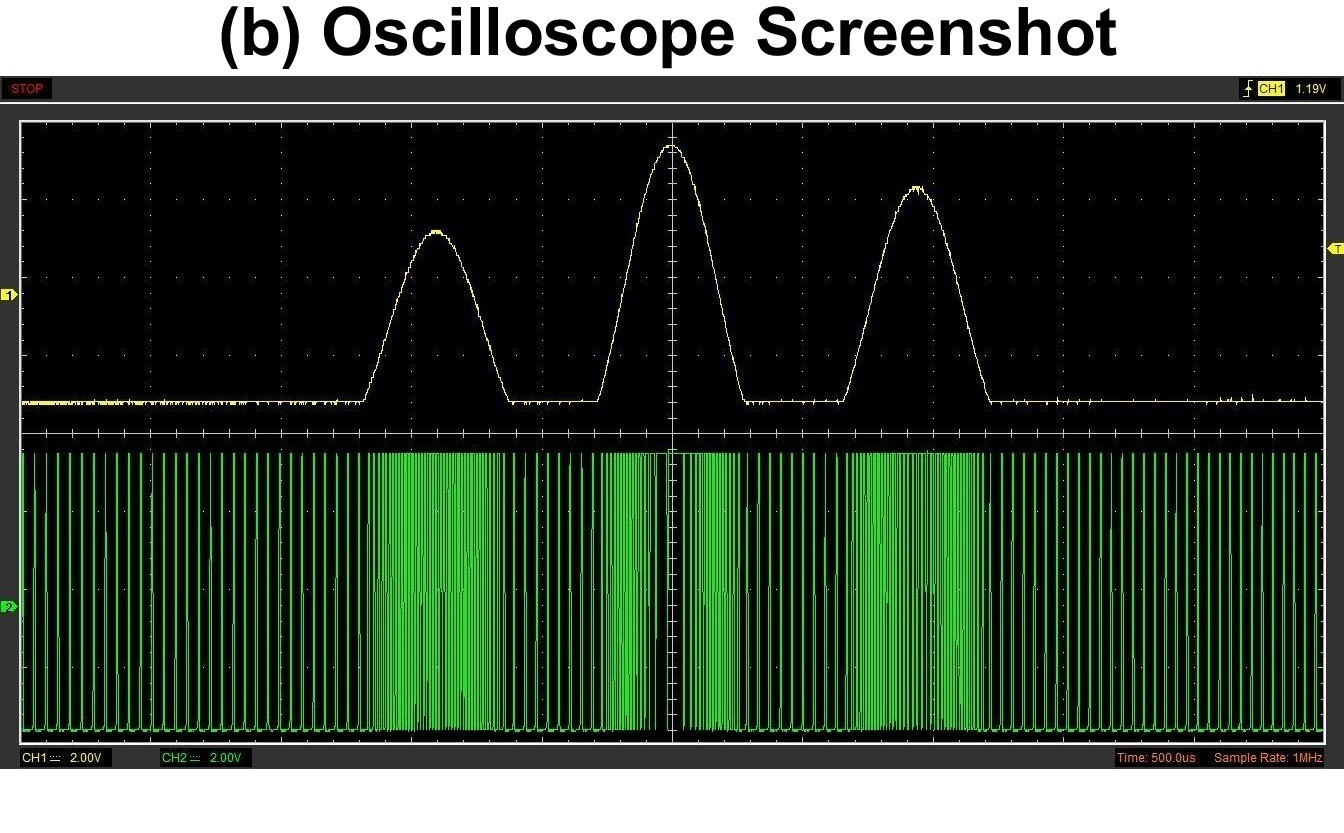}
\includegraphics[trim={0cm 0cm 0cm 0},clip,width=0.34\textwidth]{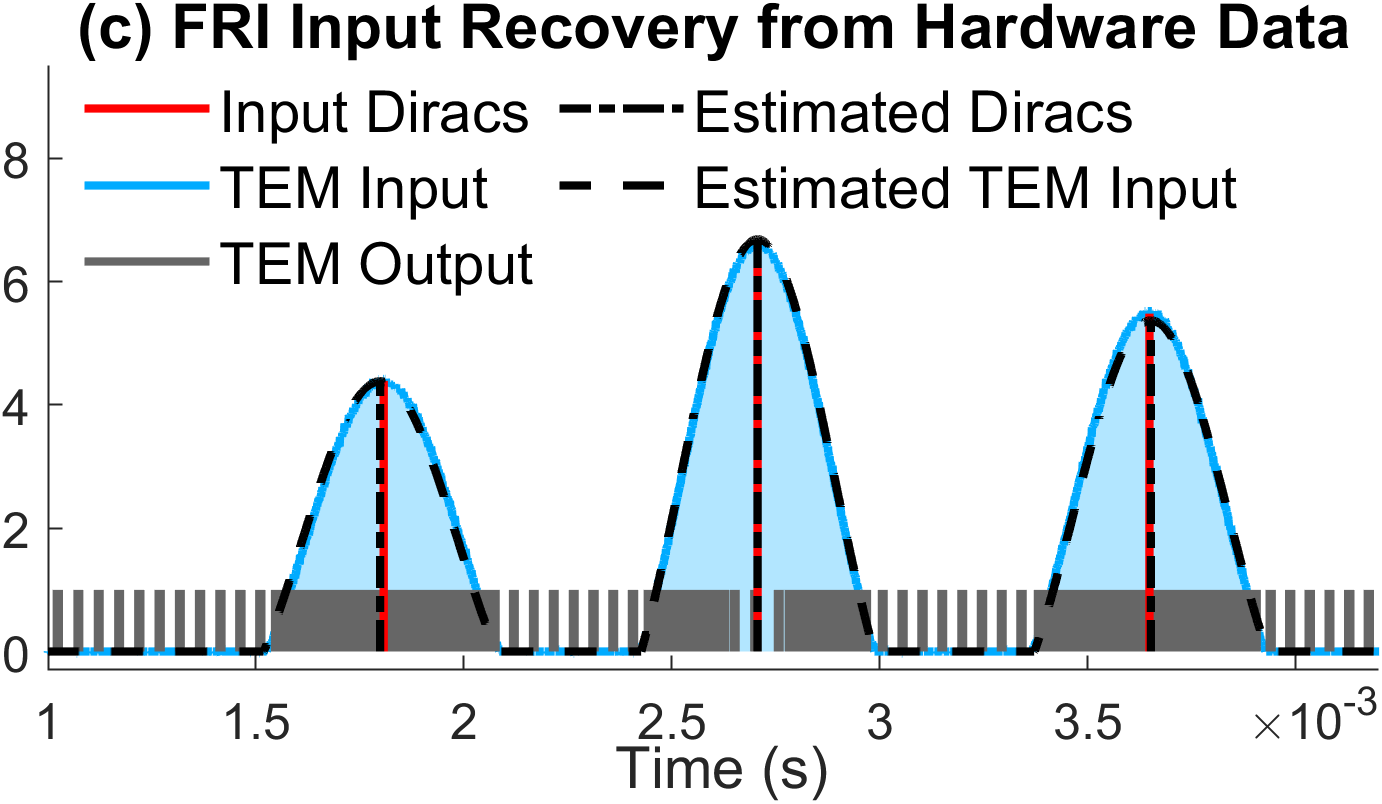}
\caption{(a) ASDM Hardware. (b) Oscilloscope screenshot depicting the ASDM input $g(t)$ via channel 1 (yellow) and ASDM output $z(t)$ via channel 2 (green).  (c) FRI Input Recovery with Algorithm \ref{alg:1}.}
\label{fig:hw}
\end{figure*}

We generated an input $g(t)$ using  $\varphi(t)=\frac{\sin(\Omega t)}{\Omega t}\cdot\ind_{\sqb{-\frac{\pi}{\Omega},\frac{\pi}{\Omega}}}(t), $ 
where $\Omega=\frac{\pi}{280}\ \mathrm{Mrad/s}$, representing the windowed main lobe of a sinc function. The input satisfies $g(t)=\sum_{k=1}^3 a_k\varphi(t-\tau_k), t\in\sqb{0,6.2\ \mathrm{ms}}$, where $\tau_1=1.807\ \mathrm{m s}$, $\tau_2=2.706\ \mathrm{m s}$, $\tau_3=3.648\ \mathrm{m s}$, $a_1=4.39$, $a_2=6.62$, $a_3=5.48$. We note that, although using the $\mathrm{sinc}$ function, signal $g(t)$ is not bandlimited due to windowing. The ASDM responded to $g(t)$ with output signal $z(t)$. 
We extracted the output switching times $\cb{t_n}_{n=1}^{417}$ by computing the zero crossings of $z(t)$. 

When the TEM is simulated, as in sections \ref{subsect:bspline}-\ref{sect:comparative}, its parameters are known \emph{a priori}. In the case of a hardware experiment, the parameters need to be identified from the data \cite{Florescu:2022:C}. To increase the precision of measurements we use one every $6$ ASDM samples denoted as $\overline{t}_n\triangleq t_{6n+1}$. Furthermore, we compute $\bLn g\triangleq \int_{\overline{t}_n}^{\overline{t}_{n+1}}g(s) ds= \sum_{m=1}^6 \LN{6n+m} g$. Using this notation, we derive the following from \eqref{eq:ttransform} 
\begin{equation*}
    \begin{aligned}
        \bLn g= \sum_{m=1}^6 \rb{-1}^{6n+m}\sqb{2\delta- b\Delta t_{6n+m}}-g_0 \Delta t_{6n+m}
        =c_1\sum_{m_1=0}^2 \Delta t_{6n+2m_1+1}-c_2\sum_{m_1=0}^2\Delta t_{6n+2m_1+2}        ,
    \end{aligned}
\end{equation*}
where $c_1=b-g_0$, $c_2=b+g_0$ represent hardware parameters.
We identify $\widehat{c}_1=7.5415$ and $\widehat{c}_2=-1.7565$ above via least squares using $\cb{\bLn g}_{41}^{53}$, which represent input integrals over the support of the last pulse centered in $\tau_3$. Subsequently, we use $\widehat{c}_1$ and $\widehat{c}_2$ to compute $\cb{\bLn g}_{1}^{69}$, which covers the whole support of $g(t)$. To compensate for nonidealities, for each pulse $k$, we compute $\widehat{n}_k$ as in Algorithm \ref{alg:1} from $\cb{\bLn g}_{1}^{69}$ using $tol=130$, and subsequently run the algorithm again by replacing step 2a) with $\widehat{n}_k^*=\widehat{n}_k+1$. Each pair of resulted values for $\widehat{\tau}_k$ and $\widehat{a}_k$ are averaged to compute the final FRI parameters. The FRI signal $\widehat{g}(t)$ and parameters $\cb{\widehat{\tau}_k,\widehat{a}_k}_{k=1}^3$ are depicted in \fig{fig:hw}(b). The corresponding recovery errors are $\mathsf{Err}_{\tau}=0.19\%$, $\mathsf{Err}_{a}=1.08\%$, and $\mathsf{Err}_g=3.15\%$.

\vspace{2em}
\section{Proofs}
\label{sect:proofs}
\if\BoundDerInAppendix\FL
\else
\begin{proof}[\textbf{Proof for Lemma \ref{lem:monotonic}}]

\end{proof}
\fi

\if\NoisyRecInAppendix\FL
\else
\begin{proof}[\textbf{Proof for Theorem \ref{theo:2} (Noisy Input Recovery)}]

\end{proof}

\fi

\vspace{2em}
\section{Conclusions}	
\label{sect:conclusions}
In this paper, we introduced a new recovery method for FRI signals from TEM measurements that can tackle a wider class of FRI filters than previously possible. We introduced guarantees in the noiseless and noisy scenarios. We validated the method numerically, showing it can tackle existing FRI filters, but also random filters, which are not compatible with existing approaches. We further validated our method via a TEM hardware experiment. When the FRI filter is not designed, as it may result from the environment and the physical properties of the acquisition device, the proposed method is still applicable and allows bypassing the filter modelling stage. Additionally, by allowing a wider class of filters, the proposed algorithm can incorporate non-idealities, thus enabling a co-design of hardware and algorithms for future FRI acquisition systems.

\ifCLASSOPTIONcaptionsoff
\newpage
\fi

\bibliographystyle{IEEEtran}
\bibliography{IEEEabrv.bib,Ref2023.bib}

\end{document}